\documentclass[letter, 11pt]{article}
\usepackage{amsmath} 
\usepackage{amssymb}  
\usepackage{amsthm}
\usepackage[dvips]{graphicx}
\usepackage{times}
\usepackage{multirow}
\newcommand{\Order}{\mathrm{O}}

\newcommand{\defeq}{\stackrel{\mbox{\scriptsize{\normalfont\rmfamily def. }}}{=}}

\newcommand{\shortqed}{\hfill \mbox{$\blacksquare$} \smallskip}
\renewcommand{\Vec}[1]{\mbox{\boldmath $#1$}}

\newcommand{\MDG}{\mathcal{G}} 
\newcommand{\ME}{\mathcal{E}} 
 
\newcommand{\dsig}{\Psi_\sigma} 

\newcommand{\Z}{Z}

\newcommand{\I}{{\cal I}} 
\newcommand{\dtv}{{\cal D}_{\rm tv}}
\newcommand{\dpw}{{\cal D}_{\rm pw}}
\newcommand{\tmix}{\tau}
\newcommand{\GCD}{{\rm GCD}}

\newtheorem{theorem}{Theorem}[section]
\newtheorem{lemma}[theorem]{Lemma}
\newtheorem{corollary}[theorem]{Corollary}
\newtheorem{proposition}[theorem]{Proposition}
\newtheorem{observation}[theorem]{Observation}

\allowdisplaybreaks[2]
\title{Deterministic Random Walks for Rapidly Mixing Chains
}
\author{
 Takeharu Shiraga\footnote{
   Graduate School of Information Science and Electrical Engineering, 
   Kyushu University, Fukuoka, Japan\protect\\ 
  {\ttfamily \{takeharu.shiraga,yamauchi,kijima,mak\}@inf.kyushu-u.ac.jp}} \and 
 Yukiko Yamauchi\footnotemark[1] \and 
 Shuji Kijima\footnotemark[1] \and
 Masafumi Yamashita\footnotemark[1]
}
\begin{document}
\makeatletter
\makeatother
\maketitle
\begin{abstract}
The rotor-router model is a deterministic process analogous to a
simple random walk on a graph.
 This paper is concerned with a generalized model, {\em
functional-router model},
   which imitates a Markov chain possibly containing irrational transition probabilities.
 We investigate the discrepancy of the number of tokens at a single vertex
    between the functional-router model and its corresponding Markov chain, and 
   give an upper bound in terms of the mixing time of the Markov chain. 

\end{abstract}

\section{Introduction}
 The \textit{rotor-router model}, also known as the {\em Propp machine}, 
  is a deterministic process analogous to a random walk on a graph~\cite{PDDK96, CS06, KKM12}. 
 In the model\footnote{See Section~\ref{sec:rotor-router}, for the detail of the rotor-router model. }, 
   tokens distributed over vertices are deterministically served to neighboring vertices 
    by rotor-routers equipped on vertices, instead of traveling on the graph at random. 
 Doerr et al.~\cite{CDST07, DF09} first called the rotor-router model {\em deterministic random walk}, 
  meaning a ``derandomized, hence {\em deterministic}, version of a {\em random walk}.''

\paragraph{Single vertex discrepancy for multiple-walk.}
 Cooper and Spencer~\cite{CS06} investigated 
  the rotor-router model (with multiple tokens, in precise; {\em multiple-walk}) on $\mathbb{Z}^n$, and  
   gave an analysis on the discrepancy on a single vertex: 
 they showed a bound that $|\chi^{(t)}_v - \mu^{(t)}_v| \leq c_n$, 
  where 
   $\chi^{(t)}_v$ (resp.\ $\mu^{(t)}_v$) denotes 
   the number (resp.\ the expected number) of tokens on vertex $v \in \mathbb{Z}^n$ 
   in a rotor-router model (resp.\ in the corresponding random walk) at time $t$ 
   on the condition that $\mu^{(0)}_v = \chi^{(0)}_v$ for any $v$, and 
   $c_n$ is a constant depending only on $n$ but independent of the total number of tokens in the system. 
 Cooper et al.~\cite{CDST07} showed $c_1 \simeq 2.29$, and  
 Doerr and Friedrich~\cite{DF09} showed that 
  $c_2$ is about 7.29 or 7.83 depending on the routing rules. 
 On the other hand, 
 Cooper et al.~\cite{CDFS10} gave 
  an example of $|\chi^{(t)}_v - \mu^{(t)}_v| = \Omega(\sqrt{kt})$ on infinite $k$-regular trees, 
  the example implies that the discrepancy can get infinitely large as increasing the total number of tokens. 

 Motivated by a derandomization of Markov chains, 
 Kijima et al.~\cite{KKM12} are concerned with multiple-walks 
  on general finite multidigraphs $(V, {\cal A})$, and 
  gave a bound $|\chi^{(t)}_v - \mu^{(t)}_v| = \Order(|V||{\cal A}|)$ 
  in case that corresponding Markov chain is ergodic, reversible and lazy. 
 They also gave some examples of $|\chi^{(t)}_v - \mu^{(t)}_v|=\Omega({|\cal A}|)$. 
 Kajino et al.~\cite{KKM13} 
  sophisticated the approach by~\cite{KKM12}, and 
  gave a bound in terms of the second largest eigenvalue and eigenvectors of the corresponding Markov chain, 
  for an arbitrary irreducible finite Markov chain, 
    which may not be lazy, reversible nor aperiodic. 

 In the context of load balancing, 
  Rabani et al.~\cite{RSW98} are concerned with a deterministic algorithm 
  similar to the rotor-router model corresponding to Markov chains with {\em symmetric} transition matrices, and 
  gave a bound $\Order(\Delta \log(|V|)/(1-\lambda_*))$, where 
   $\Delta$ denotes the maximum degree of the transition diagram and 
   $\lambda_*$ denotes the second largest eigenvalue of the transition matrix.  

 For some specific finite graphs, 
  such as hypercubes and tori, 
  some bounds on the discrepancy in terms of logarithm of the size of transition diagram 
  are known. 
 For $n$-dimensional hypercube, 
  Kijima et al.~\cite{KKM12} gave a bound $\Order(n^3)$, and 
  Kajino et al.~\cite{KKM13} improved the bound to $\Order(n^2)$. 
 Recently, 
  Akbari and Berenbrink~\cite{AB13} gave a bound $\Order(n^{1.5})$, 
  using results by Friedrich et al.~\cite{FGS12}. 
 Akbari and Berenbrink~\cite{AB13} also gave a bound $\Order(1)$ 
  for constant dimensional tori.  
 Those analyses highly depend on the structures of the specific graphs, and 
  it is difficult to extend the technique to other combinatorial graphs.  
 Kijima et al.~\cite{KKM12} gave rise to a question 
  if there is a deterministic random walk 
   for {\#}P-complete problems, such as $0$-$1$ knapsack solutions, bipartite matchings, etc., 
  such that $|\chi^{(t)}_v - \mu^{(t)}_v|$ is bounded by a polynomial in the input size.  

\paragraph{Other topics on deterministic random walk.}
 As a highly related topic, 
 Holroyd and Propp~\cite{HP10} analyzed 
   ``hitting time'' of the rotor-router model 
   with a {\em single token} ({\em single-walk}) on finite simple graphs, and  
  gave a bound $|\nu^{(t)}_v - t \pi_v| = \Order(|V||{\cal A}|)$ 
  where $\nu^{(t)}_v$ denotes the frequency of visits of the token at vertex $v$ in $t$ steps, and 
  $\pi$ denotes the stationary distribution of the corresponding random walk. 
 Friedrich and Sauerwald~\cite{FS10} 
   studied the cover time of a single-walk version of the rotor-router model 
   for several basic finite graphs such as tree, star, torus, hypercube and complete graph. 
Recently, Kosowski and Pajak~\cite{KP14} studied the cover time of a multiple tokens version of the rotor-router model. 

 Holroyd and Propp~\cite{HP10} also proposed a generalized model called {\em stack walk}, 
   which is the first model of deterministic random walk for {\em irrational} transition probabilities, as far as we know. 
 While Holroyd and Propp~\cite{HP10} indicated the existence of routers 
   approximating irrational transition probabilities well, 
 Angel et al.~\cite{AJJ10} gave a routing algorithm based on the ``shortest remaining time (SRT)'' rule. 
 Shiraga et al.~\cite{shiraga}, that is a preliminary work of this paper, 
  independently proposed another model based on the van der Corput sequence, 
  motivated by irrational transition probabilities, too. 

As another topic on the rotor-router model, 
  the aggregation model has been investigated~\cite{LP05, Kleber05, LP08, LP09}. 
 For a random walk, 
  tokens in the Internal Diffusion-Limited Aggregation (IDLA) model on $\mathbb{Z}^n$ 
   asymptotically converge to the Euclidean ball~\cite{LBG92}, and  
 Jerison, Levine and Sheffield~\cite{JLS12} recently showed  
  the fluctuations from circularity are $\Order(\log t)$ after $t$ steps.
 For the rotor-router model, 
  Levine and Peres~\cite{LP05, LP08, LP09} showed that 
  tokens in the rotor-router aggregation model also form the Euclidean ball, and  
  showed several bounds for the fluctuations. 
 Kleber~\cite{Kleber05} gave some computational results. 

 Doerr et al.~\cite{DFS08} showed that 
  information spreading by the rotor-router model is faster than the one by a random walk 
 on some specific graphs, 
  namely trees with the common depth and the common degree of inner vertices, and 
  random graphs with restricted connectivity.  
 Doerr et al.~\cite{DFKS09} 
   gave some computational results for this phenomena. 
 There is much other research on information spreading 
  by the rotor-router model on some graphs~\cite{ADHP09, Doerr09, DFS09a, DFS09b, DHL09, HF09}.

\paragraph{Our Results.}
 This paper is concerned with the {\em functional-router} model (of multiple-walk ver.), 
   which is a generalization of the rotor-router model. 
 While the rotor-router model is an analogy with a simple random walk on a graph, 
  the functional-router model imitates a Markov chain 
  possibly containing irrational transition probabilities. 
 In the functional-router model, 
  a configuration of $M$ tokens over a finite set $V = \{1, \ldots, N\}$ 
  is deterministically updated by functional-routers defined on vertices\footnote{
    See Section~\ref{sec:fr-model}, for the detail of the functional-router model.}. 
 Let $\chi^{(t)} = (\chi^{(t)}_1, \ldots, \chi^{(t)}_N) \in \mathbb{Z}_{\geq 0}^N$ denote 
   the configuration at time $t = 0, 1, 2, \ldots$, 
   i.e., $\sum_{v \in V} \chi^{(t)}_v = M$. 
 For comparison, 
  let $\mu^{(0)} = \chi^{(0)}$, and 
  let $\mu^{(t)} = \mu^{(0)}P^t$ for a transition matrix $P$ corresponding to the functional-router model, 
 then 
  $\mu^{(t)}\in \mathbb{R}_{\geq 0}^N$ denotes the expected configuration of $M$ tokens 
  independently according to $P$ for $t$ steps. 
 A main contribution of the paper is to show that  
   $|\chi^{(t)}_v - \mu^{(t)}_v| \leq 6 (\pi_{\max}/\pi_{\min}) t^* \Delta$ 
   holds for any $v \in V$ at any time $t$ 
   in case that the corresponding transition matrix $P$ is {\em ergodic} and {\em reversible}, 
  where 
  $\pi_{\max}$ and $\pi_{\min}$ are respectively the maximum/minimum values of the stationary distribution $\pi$ of $P$, 
  $t^*$ is the {\em mixing rate} of $P$, and 
  $\Delta$ is the maximum degree of the transition diagram. 

 An example of a random walk containing irrational transition probabilities is 
  the $\beta$-random walk devised by Ikeda et al.~\cite{IKY09}, 
  which achieves an $\Order(N^2)$ hitting time and an $\Order(N^2 \log N)$ cover time {\em for any graphs}. 
 Another example should be the Markov chain Monte Carlo (MCMC), 
   such as Gibbs samplers for the Ising model (cf.\ \cite{Sinclair93, PW96}), 
   reversible Markov chains for queueing networks~(cf.~\cite{KM08}), etc.

\paragraph{Organization}
 This paper is organized as follows.  
 In Section~\ref{sec:MCMC}, 
   we briefly review MCMC, 
    as a preliminary. 
 In Section~\ref{sec:modelresult}, 
  we describe the functional-router model and our main theorem.   
 In Section~\ref{sec:main}, 
   we prove the main theorem. 
 In Section~\ref{sec:routingmodel}, 
   we present four particular functional-router models, and give detailed analyses on them. 
 In Section~\ref{sec:applications}, 
   we show some examples of the bounds 
    for some Markov chains over the combinatorial objects, 
   which are known to be rapidly mixing. 

\section{Preliminaries: Markov Chain Monte Carlo}\label{sec:MCMC}
 As a preliminary step of explaining the functional-router model, 
  this section briefly reviews the Markov chain Monte Carlo (MCMC). 
 See e.g., \cite{Sinclair93, LPW08, MT06} for details of MCMC. 

 Let $V \defeq \{1, \ldots, N\}$ be a finite set, and 
  suppose that we wish to sample from $V$ with a probability 
  proportional to a given positive vector 
  $f=(f_1, \ldots, f_N) \in \mathbb{R}_{\geq 0}^{N}$; 
  for example, we are concerned with {\em uniform} sampling of $0$-$1$ knapsack solutions in Section~\ref{sec:knapsack}, 
   where $V$ denotes the set of $0$-$1$ knapsack solutions and $f_v = 1$ for each $v \in V$. 
 The idea of a Markov chain Monte Carlo (MCMC) is 
   to sample from a limit distribution of a Markov chain 
   which is equal to the target distribution $f/\|f\|_1$ 
   where $\|f\|_1 = \sum_{v \in V} f_v$ is the normalizing constant. 

 Let $P \in \mathbb{R}_{\geq 0}^{N \times N}$ be 
   a transition matrix of a Markov chain with the state space $V$, 
  where $P_{u, v}$ denotes the transition probability from $u$ to $v$ ($u, v \in V$). 
 A transition matrix $P$ is {\em irreducible} if $P^t_{u, v} > 0$ for any $u$ and $v$ in $V$, and 
  is {\em aperiodic} if $\GCD\{ t \in \mathbb{Z}_{>0} \mid P^t_{x, x} > 0\} = 1$ holds for any $x \in V$, 
  where $P^t_{u, v}$ denotes the $(u, v)$ entry of $P^t$, the $t$-th power of $P$. 
 An irreducible and aperiodic transition matrix is called {\em ergodic}. 
 It is well-known for a ergodic $P$, 
  there is a unique {\em stationary distribution} $\pi \in \mathbb{R}_{\geq 0}^{N}$, 
   i.e., $\pi P = \pi$, 
  and the limit distribution is $\pi$, 
   i.e., $\xi P^{\infty} = \pi$ for any probability distribution $\xi\in \mathbb{R}_{\geq 0}^{N}$ on $V$. 

 An ergodic Markov chain defined by a transition matrix $P \in \mathbb{R}_{\geq 0}^{N \times N}$ is {\em reversible} 
  if the {\em detailed balance equation} 
\begin{eqnarray}\label{eq:db}
  f_u P_{u, v} = f_v P_{v, u}
\end{eqnarray}
  holds for any $u, v \in V$. 
 When $P$ satisfies the detailed balance equation, 
  it is not difficult to see that $fP=f$ holds, 
  meaning that $f/\|f\|_1$ is the limit distribution (see e.g., ~\cite{LPW08}). 
Let $\xi$ and $\zeta$ be a distribution on $V$, 
then the {\em total variation distance} $\dtv$ between $\xi$ and $\zeta$ is defined  by  
\begin{eqnarray}
\label{def:TV}
\dtv(\xi, \zeta)
\defeq \max_{A\subset V} \left| \sum_{v\in A}(\xi_v-\zeta_v )\right|
=\frac{1}{2} \left\|\xi-\zeta\right\|_1. 
\label{def:dtv}
\end{eqnarray}
 Note that $\dtv(\xi, \zeta) \leq 1$, since $\|\xi\|_1$ and $\|\zeta\|_1$ are equal to one, respectively. 
The {\em mixing time} of a Markov chain is defined by 
\begin{eqnarray}
 \tau(\varepsilon) \defeq 
 \max_{v \in V} \min \left\{ t \in \mathbb{Z}_{\geq 0} \mid \dtv(P^t_{v, \cdot}, \pi) \leq \varepsilon \right\}
\label{def:mix}
\end{eqnarray}
 for any $\varepsilon > 0$, 
  where  $P^t_{v, \cdot}$ denotes the $v$-th row vector of $P^t$; 
  i.e., $P^t_{v, \cdot}$ denotes the distribution of 
   a Markov chain at time $t$ 
   stating from the initial state $v \in V$. 
 In other words, 
   the distribution $P^t_{v, \cdot}$ of the Markov chain after $\tau(\varepsilon)$ transition 
   satisfies $\dtv(P^t_{v, \cdot}, \pi) \leq \varepsilon$, 
  meaning that we obtain an approximate sample from the target distribution. 

For convenience, let
 $h(t) \defeq \max_{w\in V}\dtv\left( P^t_{w, \cdot}, \pi \right)$ 
 for $t \geq 0$, then 
 it is well-known that 
  $h$ satisfies a kind of {\em submultiplicativity}. 
 We will use the following proposition in our analysis in Section~\ref{sec:main}. 
 See Appendix~\ref{sec:RMC} for the proof (cf. \cite{LPW08,MT06}). 
\begin{proposition}\label{prop:dltimes}
  For any integers $\ell$  $(\ell\geq 1)$ and 
  $k$ $(0 \leq k < \tau(\gamma))$, 
\begin{eqnarray*}
h \left(\ell \cdotp \tau(\gamma)+k \right)\leq \frac{1}{2}(2\gamma )^\ell
\end{eqnarray*}
 holds for any $\gamma$ $(0<\gamma<1/2)$. 
\shortqed
\end{proposition} 
Since the submultiplicativity, $t^* \defeq \tau(1/4)$, called {\em mixing rate}, is often used as a characterization of $P$. 
%

\section{Model and Main Results}\label{sec:modelresult}
A {\em functional-router model} is 
   a  deterministic process analogous to a multiple random walk. 
 Roughly speaking, 
  a router defined on each vertex $u$ deterministically serves tokens to $v$ 
  at a rate of $P_{u, v}$ in a functional-router model, 
  while tokens on a vertex $u$ moves to a neighboring vertex $v$ with probability $P_{u, v}$ in a (multiple) random walk.

 To get the idea, 
  let us start 
   with explaining the rotor-router model (see e.g., ~\cite{CS06, KKM12}), 
   which corresponds to a simple random walk on a graph. 

\subsection{Rotor-router model}\label{sec:rotor-router}
 Let $\MDG=(V, \ME)$ be a simple undirected graph\footnote{
  In Section~\ref{sec:roter}, we are concerned with a model on multidigraphs. 
}, where $V=\{1, \ldots, N\}$. 
 Let ${\cal N}(v)$ denote 
  the neighborhood of $v \in V$. 
 For convenience, 
  let $\delta(v)=|{\cal N}(v)|$. 
 Let $\chi^{(0)} \in \mathbb{Z}_{\geq 0}^{N}$ be an initial configuration of tokens, and  
 let $\chi^{(t)} \in \mathbb{Z}_{\geq 0}^{N}$ denote the configuration of tokens 
   at time~$t \in \mathbb{Z}_{\geq 0}$ in the rotor-router model. 
 A configuration $\chi^{(t)}$ 
   is updated by {\em rotor-routers} on vertices, as follows. 
 Without loss of generality, 
   we may assume that an ordering $u_0, \ldots, u_{\delta(v)-1}$ is defined on ${\cal N}(v)$ for each $v \in V$. 
 Then, a rotor-router $\sigma_v \colon \mathbb{Z}_{\geq 0} \to {\cal N}(v)$ on $v \in V$ is defined by 
\begin{eqnarray}
\sigma_v(j) \defeq u_{i \bmod \delta(v)}
\label{def]rotor}
\end{eqnarray}
 for $j \in \mathbb{Z}_{\geq 0}$. 
 Let
\begin{eqnarray*}
 Z_{v, u}^{(t)} \defeq 
 \left| \left\{j \in \{0, \ldots, \chi^{(t)}_{v}-1 \} 
 \mid  \sigma_v \left(  j + \textstyle\sum_{s=0}^{t-1} \chi^{(s)}_v \right) = u \right\} \right|
\end{eqnarray*}
  for $v, u \in V$, 
 where $Z_{v, u}^{(t)}$ denotes the number of tokens served from $v$ to $u$ in the update. 
 Then, $\chi^{(t+1)}$ is defined by 
\begin{eqnarray*}\textstyle
 \chi_u^{(t+1)} 
  \defeq \sum_{v \in V} Z_{v, u}^{(t)}
\end{eqnarray*}
  for each $u \in V$. 

 It is not difficult to see that 
\begin{eqnarray*}
 \frac{| \{j \in \{0, \ldots, z-1 \} \mid  \sigma_v(j) = u \} |}{z} 
&\xrightarrow{ z \to \infty}&
 \frac{1}{\delta(v)}
\end{eqnarray*}
  holds, which implies that  
  the ``outflow ratio'' 
   $\sum_{s=0}^t Z_{v, u}^{(s)} / \sum_{s=0}^t \chi^{(s)}_v$ of tokens at $v$ to $u$ 
 approaches asymptotically to $1/\delta(v)$ as $t$ increasing. 
 Thus, the rotor-router hopefully approximates a distribution of tokens by a random walk. 

\subsection{Functional-router model}\label{sec:fr-model}
 Let $P \in \mathbb{R}_{\geq 0}^{N \times N}$ be 
   a transition matrix of a Markov chain with a state space $V \defeq \{1,\ldots,N\}$, 
  where $P_{u,v}$ denotes the transition probability from $u$ to $v$ ($u, v \in V$). 
 Note that $P_{u,v}$ may be irrational\footnote{  
     e.g., $P_{u,v} = \sqrt{5}/10$, $\exp(-10)$, $\sin(\pi/3)$, etc.~are allowed.
  }.
 In this paper, we assume that $P$ is {\em ergodic} and {\em reversible} (see Section~\ref{sec:MCMC}). 
 Let $\mu^{(0)} = (\mu^{(0)}_1, \ldots, \mu^{(0)}_N) \in \mathbb{Z}_{\geq 0}^N$ 
   denote an initial configuration of $M$ tokens over $V$, and 
 let $\mu^{(t)} \in \mathbb{R}_{\geq 0}^N$ 
   denote the {\em expected} configuration of tokens 
   independently according to $P$ at time $t \in \mathbb{Z}_{\geq 0}$, 
   i.e., $\|\mu^{(t)}\|_1 = M$ and $ \mu^{(t)} =  \mu^{(0)} P^t$. 

 Let $\MDG=(V, \ME)$ be the transition digram of~$P$, 
  meaning that $\ME = \{(u, v) \in V^2 \mid P_{u, v} > 0\}$. 
 Note that $\ME$ may contain self-loop edges, and also 
 note that $|\ME| \leq N^2$ holds. 
 Let ${\cal N}(v)$ denote 
  the (out-)neighborhood\footnote{
 Since $P$ is reversible, 
  $u \in {\cal N}(v)$ if and only if $v \in {\cal N}(u)$, and then 
  we abuse ${\cal N}(v)$ for in-neighborhood of $v \in V$. 
} of $v \in V$, 
  i.e., ${\cal N}(v) = \{ u \in V \mid P_{v, u}>0 \}$, and  
 let $\delta(v)=|{\cal N}(v)|$. 
 Note that $v \in {\cal N}(v)$ if $P_{v,v} > 0$. 

 Let $\chi^{(0)}=\mu^{(0)}$, and 
 let $\chi^{(t)} \in \mathbb{Z}_{\geq 0}^N$ denote the configuration of tokens 
   at time~$t \in \mathbb{Z}_{\geq 0}$ in the functional-router model. 
 A configuration $\chi^{(t)}$ 
   is updated by {\em functional-routers} 
   $\sigma_v \colon \mathbb{Z}_{\geq 0} \to {\cal N}(v)$ defined on each $v \in V$ 
  to imitate $P_{v, u}$. 
To be precise, let
\begin{eqnarray}
 \I_{v, u}[z, z') \defeq \left|\left\{ j \in \{z, \ldots, z'-1\} \mid \sigma_v(j)=u \right\} \right|
 \label{def:Ivu}
\end{eqnarray}
  for $v, u \in V$ and for any $z, z' \in \mathbb{Z}_{\geq 0}$ satisfying $z <z'$, 
  for convenience. 
 Then, the functional router $\sigma_v$ on $v\in V$ is designed to minimize 
\begin{eqnarray*}
\left|\frac{\I_{v, u}[0, z)}{z} - P_{v,u} \right|
\end{eqnarray*}
 for $z \in \mathbb{Z}_{\geq 0}$. 
See Section~\ref{sec:routingmodel} for some specific functional-routers. 
 Let
\begin{eqnarray}\textstyle
 \Z_{v, u}^{(t)} = \I_{v, u}\left[ \sum_{s=0}^{t-1} \chi_v^{(s)}, \sum_{s=0}^{t} \chi_v^{(s)} \right) 
 \label{def:Zvut}
\end{eqnarray}
  for $v, u \in V$, 
 where $Z_{v, u}^{(t)}$ denotes the number of tokens served from $v$ to $u$ in the update. 
 Then, $\chi^{(t+1)}$ is defined by 
\begin{eqnarray}\textstyle
 \chi_u^{(t+1)} \defeq \sum_{v \in V} Z_{v, u}^{(t)}
 \label{eq:Zvut-vsum}
\end{eqnarray}
 for each $u \in V$. 

 We in Section~\ref{sec:routingmodel} give some specific functional-routers, 
  in which the ``outflow ratio'' 
   $\sum_{s=0}^t Z_{v, u}^{(s)} / \sum_{s=0}^t \chi^{(s)}_v$ from $v$ to $u$ 
 approaches asymptotically to $P_{v, u}$ as $t$ increases, 
 meaning that the functional-router hopefully approximate a distribution of tokens by a random walk. 


 Figure~\ref{fig:one} shows an example of the time evolution of a functional router model. 
 In the example, $V=\{1, 2\}$ and the initial configuration of tokens is $\chi^{(0)}=(7, 0)$. 
 According to the functional router $\sigma_1$ defined in the figure, 
\begin{eqnarray*}
  \I_{1, 1}[0, 7) &=& \left|\left\{ j \in \{0, \ldots, 6\} \mid \sigma_1(j)=1 \right\} \right|=4 \quad \mbox{and} \\
  \I_{1, 2}[0, 7) &=& \left|\left\{ j \in \{0, \ldots, 6\} \mid \sigma_1(j)=2 \right\} \right|=3, 
\end{eqnarray*}
  and then the configuration of tokens is $\chi^{(1)}=(4, 3)$ at time 1.  
 In a similar way, 
\begin{eqnarray*}
  \I_{1, 1}[7, 11) &=& \left|\left\{ j \in \{7, \ldots, 10\} \mid \sigma_1(j)=1 \right\} \right|=3, \\
  \I_{1, 2}[7, 11) &=& \left|\left\{ j \in \{7, \ldots, 10\} \mid \sigma_1(j)=2 \right\} \right|=1, \\
  \I_{2, 1}[0, 3) &=& \left|\left\{ j \in \{ 0,1,2 \} \mid \sigma_1(j)=1 \right\} \right|=2, \quad \mbox{and} \\
  \I_{2, 2}[7, 3) &=& \left|\left\{ j \in \{ 0,1,2 \} \mid \sigma_1(j)=2 \right\} \right|=1, 
\end{eqnarray*}
  provides $\chi^{(2)}=(5, 2)$. 
\begin{figure}[t]
 \begin{center}
  \includegraphics[width=16.6cm]{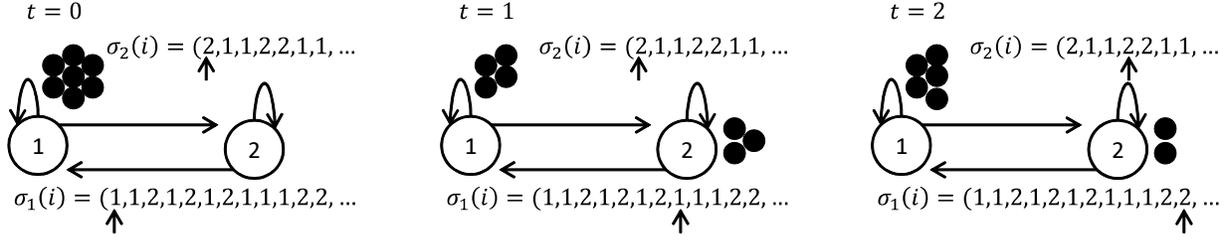}
 \end{center}
 \caption{An example of a functional router model}
 \label{fig:one}
\end{figure}

\subsection{Main results}\label{subsec:main_results}
%
Our goal is to estimate the discrepancy $|\chi^{(T)}_w - \mu^{(T)}_w |$ 
for $w \in V$ and $T \geq 0$ for the functional router model described in Section \ref{sec:fr-model}. 
%
%
For convenience, let 
\begin{eqnarray}
\label{def:dsigma}
\dsig = \max_{\substack{v\in V, \, u \in {\cal N}(v), \, t \geq 0}}\left| \Z_{v, u}^{(t)}-\chi_v^{(t)} P_{v, u}\right|
\end{eqnarray}
depending on a functional-router model $\sigma$, then the following is our main theorem. 
\begin{theorem}
\label{thm:mixupper-vertexds}
 Let $P \in \mathbb{R}_{\geq 0}^{N \times N}$ be a transition matrix of 
  a {\em reversible} and ergodic Markov chain with a state space $V$, where $\pi$ denotes the stationary distribution of $P$ and $\tau(\gamma)$ denotes the mixing time of $P$ for any $\gamma\ (0<\gamma<1/2)$. 
 Then, the discrepancy between $\chi^{(T)}$ and $\mu^{(T)}$ satisfies
\begin{eqnarray*}
\left|\chi_w^{(T)}-\mu_w^{(T)}\right|
\leq \dsig \frac{2(1-\gamma )}{1-2\gamma }\tau(\gamma )\frac{\pi_w}{\pi_{\min}}\Delta 
\end{eqnarray*}
for any $w\in V$, $T\geq 0$ and $\gamma\ (0 < \gamma < 1/2)$, where $\Delta$ denotes the maximum degree of the transition diagram of $P$, {\rm i.e.} $\Delta=\max_{v\in V}\delta(v)$. 
\end{theorem}
We remark that 
\begin{eqnarray}
\label{eq:dsigmabound}
 \dsig 
 \leq 
 \max_{\substack{v\in V, \ u\in {\cal N}(v), \\ z, z' \in \mathbb{Z}_{\geq 0} \, \mbox{s. t. }\, z'>z }}
 \left| \I_{v, u}[z, z')-(z'-z)  P_{v,u}\right|
\end{eqnarray}
 holds, 
 since 
\begin{eqnarray*}\textstyle
 \Z_{v, u}^{(t)}-\chi_v^{(t)} P_{v,u} 
 = \textstyle
  \I_{v, u}\left[ \sum_{s=0}^{t-1} \chi_v^{(s)}, \sum_{s=0}^{t} \chi_v^{(s)} \right) 
- \left(\sum_{s=0}^{t} \chi_v^{(s)} - \sum_{s=0}^{t-1} \chi_v^{(s)} \right) P_{v,u}
\end{eqnarray*}
 holds by the definition. 
 For instance, the {\em SRT router}, which we will introduce in Section~\ref{sec:greedy}, satisfies $\dsig\leq 2$, and we obtain the following, from Theorem~\ref{thm:mixupper-vertexds}. 
 %
\begin{theorem}
\label{thm:mixupper-vertexgreedy}
 Let $P \in \mathbb{R}_{\geq 0}^{N \times N}$ be a transition matrix of 
  a reversible and ergodic Markov chain with a state space $V$, where $\pi$ denotes the stationary distribution of $P$ and $t^*$ denotes the mixing rate of $P$. 
For a SRT router model, the discrepancy between $\chi^{(T)}$ and $\mu^{(T)}$ satisfies
\begin{eqnarray*}
\left|\chi_w^{(T)}-\mu_w^{(T)}\right|
\leq \frac{6\pi_w}{\pi_{\min}} t^* \Delta
\end{eqnarray*}
for any $w\in V$ and $T\geq 0$, where $\Delta$ denotes the maximum degree of the transition diagram of $P$, {\rm i.e.} $\Delta=\max_{v\in V}\delta(v)$. 
\end{theorem}

 See Section~\ref{sec:routingmodel} for detailed arguments on the bounds of $\dsig$ for some specific functional routers. 

\section{Analysis of the Point-wise Distance}\label{sec:main}
 This section proves Theorem~\ref{thm:mixupper-vertexds}. 
 Our proof technique is similar to previous works~\cite{CS06, KKM12, RSW98}, in some part. 
To begin with, we establish the following key lemma. 
\begin{lemma}
\label{lemm:maindisc}
 Let $P \in \mathbb{R}_{\geq 0}^{N \times N}$ be a transition matrix of 
  a reversible and ergodic Markov chain with a state space $V$, 
  and let $\pi$ be the stationary distribution of $P$. 
 Then, 
\begin{eqnarray*}
\chi ^{(T)}_w-\mu ^{(T)}_w=\sum_{t=0}^{T-1}\sum_{u\in V}\sum_{v\in {\cal N}(u)}\left( \Z^{(t)}_{v, u}-\chi^{(t)}_v P_{v, u}\right) \left( P^{T-t-1}_{u, w}-\pi_{w} \right) 
\end{eqnarray*}
holds for any $w\in V$ and for any $T\geq 0$.
\end{lemma}
\begin{proof}
Remark that
\begin{eqnarray}
\chi_w^{(T)} - \mu_w^{(T)}
 \ =\ \left( \chi^{(T)}-\mu^{(0)}  P^T \right)_w
 \ =\ \left(\chi^{(T)} P^0-\chi^{(0)} P^T \right)_w
\label{eq:vchi-vmu}
\end{eqnarray}
 holds where the last equality follows the assumption $\chi^{(0)}=\mu^{(0)}$. 
 It is not difficult to see that 
\begin{eqnarray*}
\chi^{(T)} P^0-\chi^{(0)} P^T
&=& \left(\chi^{(T)} P^0  -\chi^{(T-1)} P^{1}\right)
   +\left(\chi^{(T-1)} P^1  -\chi^{(T-2)} P^{2}  \right)+ \nonumber\\
&&\cdots +\left(\chi^{(2)} P^{T-2}  -\chi^{(1)} P^{T-1}  \right)
  +\left(\chi^{(1)} P^{T-1}  -\chi^{(0)} P^{T} \right)\nonumber\\
&=& \sum_{t=0}^{T-1}\left( \chi^{(t+1)} P^{T-t-1} - \chi^{(t)} P^{T-t} \right)
\end{eqnarray*}
holds, thus we have
\begin{eqnarray}
 (\ref{eq:vchi-vmu})
&=& \sum_{t=0}^{T-1}\left(\left( \chi^{(t+1)} P^{T-t-1}\right)_w -\left( \chi^{(t)} P^{T-t} \right)_w\right) \nonumber\\
&=& \sum_{t=0}^{T-1}\left(\sum_{u\in V}\chi_u^{(t+1)} P^{T-t-1}_{u, w} -\sum_{u\in V}( \chi^{(t)} P)_u  P^{T-t-1}_{u, w}\right) \nonumber\\
&=& \sum_{t=0}^{T-1}\sum_{u\in V} \left( \chi_u^{(t+1) }-( \chi^{(t)} P)_u\right)  P^{T-t-1}_{u, w}.  
\label{eq:vchi-vmu2}
\end{eqnarray}
 While 
  $\sum_{u\in V} \left( \chi_u^{(t+1) }-( \chi^{(t)} P)_u\right)  P^{T-t-1}_{u, w}$ 
   in \eqref{eq:vchi-vmu2} 
   may not be $0$ in general, 
  remark that 
\begin{eqnarray*}
 \sum_{u\in V} \left( \chi_u^{(t+1) }-( \chi^{(t)} P)_u\right)
&=&\sum_{u\in V} \chi_u^{(t+1)}- \sum_{u\in V}\sum_{v\in V} \chi^{(t)}_v P_{v, u} \\
&=&\sum_{u\in V} \chi_u^{(t+1)}- \sum_{v\in V} \chi^{(t)}_v \sum_{u\in V} P_{v, u} \\
&=& M-M 
\ =\ 0
\end{eqnarray*}
 holds for any $t \geq 0$. 
Hence
\begin{eqnarray}
(\ref{eq:vchi-vmu2})
 &=& \sum_{t=0}^{T-1}\sum_{u\in V} \left( \chi_u^{(t+1) }-( \chi^{(t)} P)_u\right)P^{T-t-1}_{u, w}
    -\sum_{t=0}^{T-1}\sum_{u\in V} \left( \chi_u^{(t+1) }-( \chi^{(t)} P)_u\right)\pi_w\nonumber\\
 &=& \sum_{t=0}^{T-1}\sum_{u\in V}\left(\chi_u^{(t+1)}-(\chi^{(t)}P)_u\right) \left( P^{T-t-1}_{u, w}-\pi_w \right) 
\label{eq:vchi-vmu3}
\end{eqnarray}
 holds. 
Since $P$ is reversible, $Z_{v, u}^{(t)}=0$ for any $v \notin {\cal N}(u)$ and 
$\chi_u^{(t+1)} = \sum_{v \in V} Z_{v, u}^{(t)} = \sum_{v \in {\cal N}(u)} Z_{v, u}^{(t)}$ holds by definition \eqref{eq:Zvut-vsum}. 
Thus, 
\begin{eqnarray*}
(\ref{eq:vchi-vmu3})
&=&
 \sum_{t=0}^{T-1}\sum_{u\in V}
  \left(\sum_{v\in {\cal N}(u)} Z_{v, u}^{(t)} - \sum_{v\in {\cal N}(u)}\chi_v^{(t)}P_{v, u}\right) 
  \left( P^{T-t-1}_{u, w}-\pi_w \right) \\
&=&
 \sum_{t=0}^{T-1}\sum_{u\in V}\sum_{v\in V}
 \left( Z_{v, u}^{(t)}-\chi_v^{(t)}P_{v, u}\right) \left( P^{T-t-1}_{u, w}-\pi_w \right) 
\end{eqnarray*}
holds, and we obtain the claim.
\end{proof}
%
Now, we are concerned with {\em reversible} Markov chains, 
 and show Theorem~\ref{thm:mixupper-vertexds}. 
%
\begin{proof}[Proof of Theorem~\ref{thm:mixupper-vertexds}]
By Lemma \ref{lemm:maindisc} and \eqref{def:dsigma}, we obtain that
\begin{eqnarray}
\label{eq:mdisc1}
\left| \chi ^{(T)}_w-\mu ^{(T)}_w\right| 
 &\leq & \sum_{t=0}^{T-1} \sum_{u\in V} \sum_{v\in {\cal N}(u)}
         \left| Z_{v, u}^{(t)}-\chi_v^{(t)}P_{v, u}\right|  \left| P^{T-t-1}_{u, w}-\pi_w \right| \\ 
 &\leq & \dsig \sum_{t=0}^{T-1}\sum_{u\in V}\sum_{v\in {\cal N}(u)} \left| P^{T-t-1}_{u, w}-\pi_w \right| \nonumber \\
\label{eq:mdisc} &=& \dsig \sum_{t=0}^{T-1}\sum_{u\in V}\delta(u)\left| P^{t}_{u, w}-\pi_w \right| 
\end{eqnarray}
  holds. 
 Since $P$ is reversible, 
  $P^t_{u, w}=\frac{\pi_w}{\pi_u}P^t_{w, u}$ holds for any $w$ and $u$ in $V$ 
 (see Proposition~\ref{prop:reversible} in Appendix~\ref{sec:RMC}). 
 Thus
\begin{eqnarray}
\label{eq:mdisc3}
(\ref{eq:mdisc})
 &=&
  \dsig \sum_{t=0}^{T-1}\sum_{u\in V} \delta(u)
  \left| \frac{\pi_w}{\pi_u}\left( P^{t}_{w, u}-\pi_u\right)\right| \nonumber \\
 &\leq& 
  \dsig \Delta \frac{\pi_w}{\pi_{\min}}  \sum_{t=0}^{T-1}\sum_{u\in V} 
  \left| P^{t}_{w, u}-\pi_u \right| \nonumber \\
 &=&
  2\dsig \Delta \frac{\pi_w}{\pi_{\min}}\sum_{t=0}^{T-1}
  \dtv\left( P^t_{w, \cdot}, \pi \right) 
\label{eq:mdisc4}
\end{eqnarray}
  where the last equality follows 
  the fact that 
   $\sum_{u\in V}|P^t_{w, u}-\pi_u| = 2\dtv\left( P^t_{w, \cdot}, \pi \right)$, 
   by the definition \eqref{def:dtv} of the total variation distance. 
By Proposition \ref{prop:dltimes}, we obtain the following. 
\begin{lemma}
\label{lemm:dtsum}
For any $v\in V$ and for any $T>0$, 
\begin{eqnarray*}
 \sum_{t=0}^{T-1} \dtv\left( P^t_{v, \cdot}, \pi \right) 
 \leq \frac{1-\gamma }{1-2\gamma }\, \tau(\gamma )
\end{eqnarray*}
holds for any $\gamma$ $(0<\gamma <1/2)$. 
\end{lemma} 

\begin{proof}
 Let $h(t) = \max_{w\in V}\dtv\left( P^t_{w, \cdot}, \pi \right)$, for convenience. 
 Then, $h(t)$ is at most $1$ for any $t \geq 0$, by the definition \eqref{def:dtv} of the total variation distance. 
 By Proposition \ref{prop:dltimes}, 
\begin{eqnarray*}
 \sum_{t=0}^{T-1} \dtv\left( P^t_{w, \cdot}, \pi \right) 
 &=& \sum_{t=0}^{T-1}h(t) 
 \ \leq \ \sum_{t=0}^{\infty}h(t) 
 \ =\ 
  \sum_{\ell =0}^{\infty} \sum_{k =0}^{\tau(\gamma)-1}
  h(\ell \cdotp \tau(\gamma)+k) \\
 &=& 
  \sum_{k =0}^{\tau(\gamma)-1} h(k) + \sum_{\ell =1}^{\infty} 
  \sum_{k =0}^{\tau(\gamma)-1} h(\ell \cdotp \tau(\gamma)+k) 
\ \leq \  
  \sum_{k =0}^{\tau(\gamma)-1} 1 + \sum_{\ell =1}^{\infty} 
  \sum_{k =0}^{\tau(\gamma)-1} \frac{1}{2}\left( 2\gamma \right)^{\ell} \\
& = &
  \tau(\gamma) + \sum_{\ell =1}^{\infty} \tau(\gamma) \frac{1}{2}\left( 2\gamma \right)^{\ell} 
\ =\ 
  \tau(\gamma)+\frac{\gamma }{1-2\gamma }\, \tau(\gamma)
\ =\ 
  \frac{1-\gamma }{1-2\gamma }\, \tau(\gamma)
\end{eqnarray*}
holds, and we obtain the claim.
\end{proof}

Now we obtain Theorem~\ref{thm:mixupper-vertexds} 
from \eqref{eq:mdisc3} and Lemma~\ref{lemm:dtsum}
\end{proof}

%
\section{Specific Functional-routers}\label{sec:routingmodel}
 This section shows some functional-router models, 
   namely {\em SRT router} in Section~\ref{sec:greedy}, 
    {\em billiard router} in Section~\ref{sec:billiard}, 
   {\em quasi-random router} in Section~\ref{sec:vander}, and  
   rotor-router on multigraph  in Section~\ref{sec:roter}. 
Using Theorem~\ref{thm:mixupper-vertexds}, we give upper bounds of $|\chi^{(T)}_w - \mu^{(T)}_w|$ for them. 
%
\subsection{SRT router}\label{sec:greedy}
 This section introduces {\em SRT router}, 
  which is originally given by Holroyd and Propp~\cite{HP10} and Angel et al.~\cite{AJJ10} by the name of stack-walk. 
 The SRT router $\sigma_v(i)$ ($i \in \mathbb{Z}_{\geq 0}$) on $v \in V$ is defined, as follows. 
Let 
\begin{eqnarray} 
T_i(v)=\{u\in {\cal N}(v)\mid \I_{v,u}[0,i)-(i+1)P_{v,u}<0\}. 
\end{eqnarray}
Then, let $\sigma_v(i)$ be $u^*\in T_i(v)$ minimizing the value 
\begin{eqnarray} 
\frac{\I_{v,u}[0,i)+1}{P_{v,u}}
\end{eqnarray}
in all $u\in T_i(v)$. 
If there are two or more such $u\in T_v(i)$, then let $u^*$ be arbitrary one of them. 

 Since $\sigma_v(i) \in T_i(v)$,  
  we can see that $\I_{v, u}[0, i+1) - (i+1)P_{v, u} < 1$ holds 
  for any $u$, $v$ and $i$, 
  by an induction on $i \in \mathbb{Z}_{\geq 0}$. 
 The following theorem is 
  due to Angel et al. \cite{AJJ10} and Tijdeman~\cite{T80}.
\begin{theorem}\cite{T80, AJJ10}
\label{thm:upper-const}
 For any transition matrix $P$, 
\begin{eqnarray*}
 \left|\I_{v, u}[0, z)- z\cdotp P_{v,u}\right| < 1
\end{eqnarray*}
 holds for any $v, u\in V$ and any $z \in \mathbb{Z}_{> 0}$. 
\end{theorem}
 Theorem \ref{thm:upper-const} was firstly given by Tijdeman \cite{T80}, 
   where he gave a slightly better bound 
  $\left|\I_{v, u}[0, z)- z\cdotp P_{v,u}\right|\leq 1-(2(\delta(v)-1))^{-1}$, in fact. 
 Angel et al.~\cite{AJJ10} rediscovered Theorem \ref{thm:upper-const} 
   in the context of deterministic random walk (see also~\cite{HP10}), 
  where they also showed a similar statement holds 
  even when the corresponding probability is time-inhomogeneous.  

 Theorem~\ref{thm:upper-const} and (\ref{eq:dsigmabound}) imply that 
\begin{eqnarray}
 \dsig 
 \leq 
 \max_{\substack{v\in V, \ u\in {\cal N}(v), \\ z, z' \in \mathbb{Z}_{\geq 0} \, \mbox{s. t. }\, z'>z }}
 \left| \I_{v, u}[z, z')-(z'-z)  P_{v,u}\right|< 2
\label{bound:greedy}
\end{eqnarray}
 holds for the SRT router model.

\begin{proof}[Proof of Theorem~\ref{thm:mixupper-vertexgreedy}]
By Theorem~\ref{thm:mixupper-vertexds} and (\ref{bound:greedy}), 
\begin{eqnarray*}
\left|\chi_w^{(T)}-\mu_w^{(T)}\right|
\leq \dsig \frac{2(1-\gamma )}{1-2\gamma }\tau(\gamma )\frac{\pi_w}{\pi_{\min}}\Delta 
< 2\cdot \frac{2\cdot(1-1/4)}{1-2\cdot (1/4)}\tau(1/4) \frac{\pi_w}{\pi_{\min}}\Delta 
= \frac{6\pi_w}{\pi_{\min}} t^* \Delta
\end{eqnarray*}
holds, and we obtain the claim. 
\end{proof}

\subsection{Billiard router}\label{sec:billiard}
{\em Billiard sequence} is known to be a balanced sequence (cf. \cite{SMK04}). 
This section presents a functional router based on the billiard sequence. 

The billiard sequence is given in a similar to the SRT router, but simpler. 
Let $\sigma_v(i)$ be $u^*\in {\cal N}(v)$ minimizing the value 
\begin{eqnarray*} 
\frac{\I_{v,u}[0,i)+1}{P_{v,u}}
\end{eqnarray*}
in all $u\in {\cal N}(v)$, and if there are two or more such $u\in {\cal N}(v)$, then let $u^*$ be arbitrary one of them. 
Then, the following theorem for the billiard sequence is known. 
\begin{lemma}\label{bound:billiardz}\cite{SMK04}
For any transition matrix $P$, 
\begin{eqnarray*}
\left| \I_{v, u}[z, z')-(z'-z)P_{v, u}\right|\leq 1+(\delta(v) -2)P_{v,u}
\end{eqnarray*}
 holds for any $v, u\in V$, and for any $z, z'\in \mathbb{Z}_{\geq 0}$ satisfying $z'>z$. 
\end{lemma}
%
Using Lemma~\ref{bound:billiardz}, we obtain an upper bound of $\Psi_\sigma$ for the billiard sequence. 
%
\begin{lemma}\label{bound:billiard}
$\dsig \leq \Delta -1$ holds for the billiard sequence. 
\end{lemma}
 Thus, we obtain Theorem~\ref{thm:mixupper-vertexbi}
  by Theorem~\ref{thm:mixupper-vertexds} and Lemma~\ref{bound:billiard}. 
\begin{theorem}
\label{thm:mixupper-vertexbi}
 Let $P \in \mathbb{R}_{\geq 0}^{N \times N}$ be a transition matrix of 
  a reversible and ergodic Markov chain with a state space $V$, where $\pi$ denotes the stationary distribution of $P$ and $t^*$ denotes the mixing rate of $P$. 
  For a billiard router model, the discrepancy between $\chi^{(T)}$ and $\mu^{(T)}$ satisfies
\begin{eqnarray*}
\left|\chi_w^{(T)}-\mu_w^{(T)}\right|
\leq \frac{3\pi_w}{\pi_{\min}} t^* \Delta(\Delta-1)
\end{eqnarray*}
  for any $w\in V$ and $T\geq 0$, where $\Delta$ denotes the maximum degree of the transition diagram of $P$, {\rm i.e.} $\Delta=\max_{v\in V}\delta(v)$. 
\end{theorem}
In fact, we obtain a better bound for the billiard router model as follows, by analyzing carefully. See Appendix~\ref{appendix:billiard} for the proof. 
\begin{theorem}
\label{thm:mixupper-vertexbi2}
 Let $P \in \mathbb{R}_{\geq 0}^{N \times N}$ be a transition matrix of 
  a reversible and ergodic Markov chain with a state space $V$, where $\pi$ denotes the stationary distribution of $P$ and $t^*$ denotes the mixing rate of $P$. 
  For a billiard router model, the discrepancy between $\chi^{(T)}$ and $\mu^{(T)}$ satisfies
\begin{eqnarray*}
\left|\chi_w^{(T)}-\mu_w^{(T)}\right|
\leq \frac{6\pi_w}{\pi_{\min}} t^* (\Delta-1)
\end{eqnarray*}
  for any $w\in V$ and $T\geq 0$, where $\Delta$ denotes the maximum degree of the transition diagram of $P$, {\rm i.e.} $\Delta=\max_{v\in V}\delta(v)$. 
\end{theorem}

\subsection{Quasi-random router}\label{sec:vander}
 This section gives a router $\sigma$ 
  based on the {\em van der Corput sequence}~\cite{JGV35,N78}, 
 which is a well-known low-discrepancy sequence. 

The van der Corput sequence $\psi \colon \mathbb{Z}_{\geq 0} \to [0, 1)$ is defined as follows.  
Suppose $i \in \mathbb{Z}_{> 0}$ 
  is represented in binary as $i=\sum^{\lfloor \lg i \rfloor}_{j=0} \beta_j(i) \cdotp 2^j$ 
  using $\beta_j(i) \in \{0, 1\}$ ($j \in \{0, 1, \ldots, \lfloor \lg i \rfloor\}$). 
 Then, we define  
\begin{eqnarray}
  \psi(i) \defeq \sum^{\lfloor \lg i \rfloor}_{j=0} \beta_j(i) \cdotp 2^{-(j+1)}
\label{eq:def-psi}
\end{eqnarray}
and $\psi(0) \defeq 0$. 
 For example, 
$\psi(1) = 1 \times 1/2 = 1/2$, 
$\psi(2) = 0 \times 1/2 + 1 \times 1/4 = 1/4$, 
$\psi(3) = 1 \times 1/2 + 1 \times 1/4 = 3/4$, 
$\psi(4) = 0 \times 1/2 + 0 \times 1/4 + 1 \times 1/8 = 1/8$, 
$\psi(5) = 1 \times 1/2 + 0 \times 1/4 + 1 \times 1/8 = 5/8$, 
$\psi(6) = 0 \times 1/2 + 1 \times 1/4 + 1 \times 1/8 = 3/8$, 
  and so on. 
 Clearly, $\psi(i) \in [0, 1)$ holds for any (finite) $i \in \mathbb{Z}_{\geq 0}$. 


Now, given $i \in \mathbb{Z}_{> 0}$, we define $\sigma_v(i)$ as follows.
 Without loss of generality, 
  we may assume that an ordering $u_1, \ldots, u_{\delta(v)}$ is defined on ${\cal N}(v)$ for $v \in V$. 
 Then, 
  we define the functional-router $\sigma_v \colon \mathbb{Z}_{\geq 0} \to {\cal N}(v)$ on $v \in V$ 
  such that $\sigma_v(i) = u_k \in {\cal N}(v)$ satisfies that 
\begin{eqnarray*} \textstyle
  \sum_{j=1}^{k-1}P_{v, u_j} \leq  \psi(i) < \sum_{j=0}^{k}P_{v, u_j}  
\end{eqnarray*}
  for $k \in \{1, \ldots, \delta(v)\}$, 
 where $\sum_{j=1}^0 P_{v, u_j}=0$, for convenience. 

The following theorem is due to van der Corput \cite{JGV35}. 
\begin{theorem}\label{thm:upper-logM}\cite{JGV35}
 For any transition matrix $P$, 
\begin{eqnarray*}
 \left|\I_{v, u}[0, z) - z\cdotp P_{v, u}\right| 
 \leq \lg(z+1)
\end{eqnarray*}
 holds for any $v, u\in V$ and any $z \in \mathbb{Z}_{> 0}$. 
\end{theorem}
 More sophisticated bounds are found in~\cite{N78}. 
 Carefully examining Theorem~\ref{thm:upper-logM}, 
  we obtain the following lemma. 
  See Appendix~\ref{appendix:vander} for the proof. 
%
\begin{lemma}\label{bound:vandercz}
For any transition matrix $P$, 
\begin{eqnarray*}
\left| \I_{v, u}[z, z')-(z'-z)P_{v, u}\right|\leq 2\lg(z'-z+1)
\end{eqnarray*}
 holds for any $v, u\in V$, and for any $z, z'\in \mathbb{Z}_{\geq 0}$ satisfying $z'>z$. 
\end{lemma}
%
 Lemma~\ref{bound:vandercz} suggests the following lemma. 
%
\begin{lemma}\label{bound:vanderc}
$\dsig \leq 2\lg(M+1)$ holds for the van der Corput sequence. 
\end{lemma}

By Theorem~\ref{thm:mixupper-vertexds} and Lemma \ref{bound:vanderc}, we obtain the following. 

\begin{theorem}
\label{thm:mixupper-vertexvc}
 Let $P \in \mathbb{R}_{\geq 0}^{N \times N}$ be a transition matrix of 
  a reversible and ergodic Markov chain with a state space $V$, where $\pi$ denotes the stationary distribution of $P$ and $t^*$ denotes the mixing rate of $P$. 
  For a quasi-random router model, the discrepancy between $\chi^{(T)}$ and $\mu^{(T)}$ satisfies
\begin{eqnarray*}
\left|\chi_w^{(T)}-\mu_w^{(T)}\right|
\leq \frac{6\pi_w}{\pi_{\min}} \lg(M+1) \cdotp t^* \Delta
\end{eqnarray*}
for any $w\in V$ and $T\geq 0$, where $\Delta$ denotes the maximum degree of the transition diagram of $P$, {\rm i.e.} $\Delta=\max_{v\in V}\delta(v)$ and $M$ denotes the total number of tokens on $V$. 
\end{theorem}
 Though the bound depends on $\log M$,  
   $|\chi^{(t)}_v/M - \mu^{(t)}_v/M| = \Order(\log(M)/ M)$ holds in terms of $M$, 
   meaning that the discrepancy approaches asymptotically to zero 
     as increasing the number of tokens $M$. 

\subsection{Rotor-router on multidigraph}\label{sec:roter}
 The rotor-router model described in Section~\ref{sec:fr-model} 
   can be generally considered on digraphs with parallel edges (i.e., multidigraphs). 
 Kijima et al.~\cite{KKM12} and Kajino et al.~\cite{KKM13} are concerned with 
   the rotor-router model on finite multidigraphs. 
 Suppose that $P$ is a transition matrix with {\em rational} entries. 
 For each $v \in V$, 
  let 
  $\bar{\delta}(v) \in \mathbb{Z}_{\geq 0}$ be a common denominator (or the least common denominator) 
   of $P_{v,u}$ for all $u \in {\cal N}(v)$, 
  meaning that $\bar{\delta}(v) \cdotp P_{v,u}$ is integer for each $u\in {\cal N}(v)$. 
 We define a rotor-router 
   $\sigma_v(0), \sigma_v(1), \ldots , \sigma_v(\bar{\delta}(v)-1)$ arbitrarily 
  satisfying that 
\begin{eqnarray*}
  \left|\{ j\in[ 0, \ldots, \bar{\delta}(v)) \mid \sigma_v(j)=u \}\right| =\bar{\delta}(v)\cdotp P_{v,u}
\end{eqnarray*}
 for any $v\in V$ and $u \in {\cal N}(v)$. 
Then, $\sigma_v(i)$ is defined by 
\begin{eqnarray}
\sigma_v(i)= \sigma_v(i \bmod \bar{\delta}(v)) \left(\equiv \sigma_v\left(i - \bar{\delta}(v) \cdotp \left\lfloor \frac{i}{\bar{\delta}(v)} \right\rfloor \right) \right). 
\end{eqnarray}

 For the rotor router on a multidigraph, we have $\left \lfloor \frac{z'-z}{{\bar \delta}(v)}\right \rfloor\cdotp {\bar \delta}(v)P_{v,u}\leq \I_{v, u}[z, z')\leq \left(\left \lfloor \frac{z'-z}{{\bar \delta}(v)}\right \rfloor +1\right)\cdotp {\bar \delta}(v)P_{v,u}$, hence
  it is not difficult to observe the following. 
\begin{observation}\label{bound:rotorz}
For any transition matrix $P$, 
\begin{eqnarray*}
\left| \I_{v, u}[z, z')-(z'-z)P_{v, u}\right|\leq \bar{\delta}(v)P_{v,u}
\end{eqnarray*}
 holds for any $v, u\in V$, and for any $z, z'\in \mathbb{Z}_{\geq 0}$ satisfying $z'>z$. 
\end{observation}
Using Observation~\ref{bound:rotorz}, we obtain the following lemma. 
\begin{lemma}\label{bound:rotor}
 $\dsig =\bar{\Delta}$ holds for the rotor-router model on a multidigraph, where $\bar{\Delta}=\max_v\bar{\delta}(v)$. 
\end{lemma}

 By Theorem \ref{thm:mixupper-vertexds}, and the above lemma, 
  we obtain the following theorem. 
\begin{theorem}
\label{thm:mixupper-vertexrr}
 Let $P \in \mathbb{Q}_{\geq 0}^{N \times N}$ be a transition matrix of 
  a reversible and ergodic Markov chain with a state space $V$, where $\pi$ denotes the stationary distribution of $P$ and $t^*$ denotes the mixing rate of $P$. 
  For a rotor router model, the discrepancy between $\chi^{(T)}$ and $\mu^{(T)}$ satisfies
\begin{eqnarray*}
\left|\chi_w^{(T)}-\mu_w^{(T)}\right|
\leq \frac{3\pi_w}{\pi_{\min}} t^* \Delta \bar{\Delta}
\end{eqnarray*}
for any $w\in V$ and $T\geq 0$, where $\Delta$ denotes the maximum degree of the transition diagram of $P$, {\rm i.e.} $\Delta=\max_{v\in V}\delta(v)$, and $\bar{\Delta}=\max_v\bar{\delta}(v)$. 
\end{theorem}
Analyzing carefully, we obtain the following upper bound for the weighted rotor router model. See appendix~\ref{appendix:rotor} for the proof.  
\begin{theorem}
\label{thm:mixupper-vertexrr2}
 Let $P \in \mathbb{Q}_{\geq 0}^{N \times N}$ be a transition matrix of 
  a reversible and ergodic Markov chain with a state space $V$, where $\pi$ denotes the stationary distribution of $P$ and $t^*$ denotes the mixing rate of $P$. 
  For a rotor router model, the discrepancy between $\chi^{(T)}$ and $\mu^{(T)}$ satisfies
\begin{eqnarray*}
\left|\chi_w^{(T)}-\mu_w^{(T)}\right|
\leq \frac{3\pi_w}{\pi_{\min}}t^* \bar{\Delta}
\end{eqnarray*}
or any $w\in V$ and $T\geq 0$, where $\bar{\Delta}=\max_v\bar{\delta}(v)$. 
\end{theorem}

\section{Bounds For Rapidly Mixing Chains}\label{sec:applications}
 This section shows some examples of 
   bounds suggested by Theorems~\ref{thm:mixupper-vertexgreedy} and \ref{thm:mixupper-vertexbi2} 
  for some celebrated Markov chains known to be rapidly mixing,  
   namely ones for $0$-$1$ knapsack solutions (Section~\ref{sec:knapsack}), 
  linear extensions (Section~\ref{sec:linear_extensions}), and 
  matchings (Section~\ref{sec:matching}). 

\subsection{$0$-$1$ knapsack solutions}\label{sec:knapsack}
 Given $\Vec{a} \in \mathbb{Z}_{> 0}^n$ and $b\in \mathbb{Z}_{> 0}$, 
  the set of  $0$-$1$ knapsack solutions is defined by 
  $\Omega_{\rm Kna}=\{\Vec x \in \{0, 1\}^n \mid \sum_{i=1}^{n}a_i x_i\leq b\}$. 
 We define a transition matrix 
   $P_{\rm Kna} \in \mathbb{R}^{|\Omega_{\rm Kna}| \times |\Omega_{\rm Kna}|}$ 
  by
\begin{align}
\label{eq:knapsackP}
P_{\rm Kna}(\Vec x, \Vec y) = \left\{
\begin{array}{ll}
1/2n & (\mbox{if } \Vec y \in {\cal N}_{\rm Kna}(\Vec x))\\
1-|{\cal N}_{\rm Kna}(\Vec x)|/2n & (\mbox{if } \Vec y=\Vec x)\\
0 & (\mbox{otherwise})
\end{array}
\right. 
\end{align}
for $\Vec x, \Vec y\in \Omega_{\rm Kna}$, where ${\cal N}_{\rm Kna}(\Vec x)=\{\Vec y\in \Omega_{\rm Kna}\mid \|\Vec x-\Vec y\|_{1}=1\}$. 
Note that the stationary distribution of $P_{\rm Kna}$ is uniform distribution since $P_{\rm Kna}$ is symmetric. 
 The following theorem is due to Morris and Sinclair~\cite{MS04}. 
\begin{theorem}
\label{thm:knapsackmix}
\cite{MS04}
 The mixing time $\tau(\gamma)$ of $P_{\rm Kna}$ is 
$\Order ( n^{\frac{9}{2}+\alpha}\log \gamma^{-1} )$
 for any $\alpha>0$ and for any $\gamma>0$. 
\end{theorem}

 Thus, Theorem~\ref{thm:mixupper-vertexgreedy} (resp.\ Theorem~\ref{thm:mixupper-vertexbi2}) suggests the following. 
\begin{theorem}\label{thm:knapsackupper}
For the SRT-router model (as well as the billiard-router model) corresponding to $P_{\rm Kna}$, 
 the discrepancy between $\chi^{(T)}$ and $\mu^{(T)}$ satisfies
\begin{eqnarray*}
\left|\chi_w^{(T)}-\mu_w^{(T)}\right| = \Order ( n^{\frac{11}{2}+\alpha})
\end{eqnarray*}
for any $w\in V$ and $T\geq 0$, where $\alpha>0$ is an arbitrary constant. 
\end{theorem}

 Let  $\widetilde{\mu}^{(t)} =  \mu^{(t)}/M$, for simplicity, 
  then clearly $\widetilde{\mu}^{(\infty)} = \pi$ holds, 
  since $P$ is ergodic (see Section~\ref{sec:MCMC}).  
 By the definition of the mixing time, 
  $\dtv(\widetilde{\mu}^{(\tau(\varepsilon))},\pi) \leq \varepsilon$ holds 
 where $\tau(\varepsilon)$ denotes the mixing time of~$P$, 
 meaning that $\widetilde{\mu}$ approximates the target distribution $\pi$ well. 
 Thus, we hope for a deterministic random walk 
  that the ``distribution'' 
  $\widetilde{\chi}^{(T)} \defeq \chi^{(T)}/M$ 
  approximates the target distribution $\pi$ well. 
 For convenience, 
  a {\em point-wise distance} $\dpw(\xi,\zeta)$ 
   between $\xi \in \mathbb{R}_{\geq 0}^{N}$ and $\zeta \in \mathbb{R}_{\geq 0}^{N}$ satisfying $\|\xi\|_1 = \|\zeta\|_1 = 1$ 
  is defined by 
\begin{eqnarray}
 \dpw(\xi,\zeta) 
\defeq \max_{v \in V} |\xi_v - \zeta_v|
= \|\xi - \zeta \|_{\infty}. 
\end{eqnarray}

\begin{corollary}
 For an arbitrary $\varepsilon$ $(0<\varepsilon<1)$, 
  let the total number of tokens $M := c_1\, n^{\frac{11}{2}+\alpha}\varepsilon ^{-1}$ 
   with some appropriate constants $c_1$ and $\alpha$. 
 Then, 
  the pointwise distance between $\widetilde{\chi}^{(T)} \defeq \chi^{(T)}/M$ and $\pi$ satisfies
\begin{eqnarray}
 \dpw\left(\widetilde{\chi}^{(T)},\pi\right) \leq \varepsilon
\label{eq:knapsack}
\end{eqnarray}
 for any $T \geq c_2\, n^{\frac{9}{2}+\alpha}\log \varepsilon ^{-1}$ 
   with an appropriate constant $c_2$, 
 where $\pi$ is the uniform distribution over~$\Omega_{\rm Kna}$. 
\end{corollary}

\subsection{Linear extensions of a poset}\label{sec:linear_extensions}
Let $S=\{1, 2, \ldots, n\}$, and $Q=(S, \preceq )$ be a partial order. 
A linear extension of $Q$ is a total order $X=(S, \sqsubseteq )$ which respects $Q$, 
i.e., for all $i, j\in S$, $i\preceq j$ implies $i\sqsubseteq j$. 
Let $\Omega_{\rm Lin}$ denote the set of all linear extensions of $Q$. 
%
 We define a relationship $X \sim_p X'$ ($p \in \{1,\ldots, n\}$) 
  for a pair of linear extensions $X$ and $X'$ $\in \Omega_{\rm Lin}$ 
  satisfying that $x_p=x'_{p+1}$, $x_{p+1}=x'_p$, and $x_i=x'_i$ for all $i\neq p, p+1$, 
 i.e., 
\begin{eqnarray*}
X&=&(x_1, x_2, \ldots, x_{p-1}, x_p, x_{p+1}, x_{p+2}, \ldots , x_n) \\
X'&=&(x_1, x_2, \ldots, x_{p-1}, x_{p+1}, x_p, x_{p+2}, \ldots , x_n)
\end{eqnarray*}
holds. 
Then, we define a transition matrix $P_{\rm Lin} \in \mathbb{R}^{|\Omega_{\rm Lin}| \times |\Omega_{\rm Lin}|}$ by
\begin{align}
\label{eq:linearP}
P_{\rm Lin}(X, X') = \left\{
\begin{array}{ll}
F(p)/2 & (\mbox{if } X'\sim_p X)\\
1-\sum_{I\in {\cal N}_{\rm Lin}(X)}P_{\rm Lin}(X, I)& (\mbox{if } X'=X)\\
0 & (\mbox{otherwise})
\end{array}
\right. 
\end{align}
for $X, X'\in \Omega_{\rm Lin}$, where ${\cal N}_{\rm Lin}(X)=\{Y\in \Omega_{\rm Lin}\mid X\sim_p Y(p\in \{1, \ldots, n-1\})\}$ and $F(p)=\frac{p(n-p)}{\frac{1}{6}(n^3-n)}$. 
Note that $P_{\rm Lin}$ is ergodic and reversible, and 
 its stationary distribution is uniform on $\Omega_{\rm Lin}$~\cite{BD99}. 
 The following theorem is due to Bubley and Dyer~\cite{BD99}. 
\begin{theorem}
\label{thm:linearmix}
\cite{BD99}
For $P_{\rm Lin}$, 
\begin{eqnarray*}
\tau(\gamma)\leq \left \lceil \frac{1}{6}(n^3-n) \ln \frac{n^2}{4\gamma } \right \rceil 
\end{eqnarray*}
holds for any $\gamma>0$. 
\end{theorem}

 It is not difficult to see that the maximum degree $\Delta = n$ (including a self-loop) 
   of the transition diagram $P_{\rm LIN}$. 
 Thus, Theorem~\ref{thm:mixupper-vertexbi2} suggests the following\footnote{
   Theorem~\ref{thm:mixupper-vertexgreedy} also suggests that 
$|\chi_w^{(T)}-\mu_w^{(T)} | \leq n\left\lceil \frac{1}{3}(n^3-n)\ln n \right\rceil = \Order(n^4 \log n)$ 
for the SRT-router model. 
  }. 
\begin{theorem}\label{thm:linearupper-SRT}
For the billiard-router model corresponding to $P_{\rm LIN}$, 
 the discrepancy between $\chi^{(T)}$ and $\mu^{(T)}$ satisfies
\begin{eqnarray*}
\left|\chi_w^{(T)}-\mu_w^{(T)}\right| \leq (n-1)\left\lceil \frac{1}{3}(n^3-n)\ln n \right\rceil = \Order(n^4 \log n)
\end{eqnarray*}
for any $w\in V$ and $T\geq 0$. 
\end{theorem}

\subsection{Matchings in a graph}\label{sec:matching}
 Counting all matchings in a graph, 
   related to the {\em Hosoya index}~\cite{Hosoya71}, 
  is known to be {\#}P-complete~\cite{Valiant79b}. 
 Jerrum and Sinclair~\cite{JS96} gave a rapidly mixing chain. 
 This section is concerned with a Markov chain for sampling from all matchings in a graph\footnote{ 
 Remark that 
  counting all {\em perfect} matchings in a bipartite graph, 
   related to the {\em permanent}, 
  is also well-known {\#}P-complete problem, 
  and  
   Jerrum et al.~\cite{JSV04} gave a celebrated FPRAS 
   based on an MCMC method using annealing.
 To apply our bound to a Markov chain for sampling perfect matchings, 
  we need some assumptions on the input graph (see e.g.,~\cite{Sinclair93,JS96,JSV04}). 
}. 

Let $H=(U,F)$ be an undirected graph, where $|U|=n$ and $|F|=m$. 
A matching in $H$ is a subset ${\cal M}\subseteq F$ such that no edges in ${\cal M}$ share an endpoint. 
Let $\Omega_{\rm Mat}$ denote the set of all possible matchings of $H$. 
Let $N_C({\cal M})=\{e=\{u,v\}\mid e\notin {\cal M}, {\rm both}\ u\ {\rm and}\ v\ {\rm are\ matched\ in}\ {\cal M}\}$
and let ${\cal N}_{\rm Mat}({\cal M})=\{e\mid e\notin N_C({\cal M})\}$.
Then, for $e=\{u,v\}\in {\cal N}_{\rm Mat}({\cal M})$, we define ${\cal M}(e)$ by 
\begin{eqnarray*}
{\cal M}(e)=\left\{ \begin{array}{ll}
{\cal M}-e & ({\rm if}\ e\in {\cal M}) \\
{\cal M}+e & ({\rm if}\ u\ {\rm and}\ v\ {\rm are\ unmatched\ in}\ {\cal M}) \\
{\cal M}+e-e' & (\mbox{if exactly one of $u$ and $v$ is matched in $M$, and $e'$ is the matching edge}). 
\end{array} \right.
\end{eqnarray*}
The we define the transition matrix $P_{\rm Mat}\in \mathbb{R}^{|\Omega_{\rm Mat}| \times |\Omega_{\rm Mat}|}$ by 
\begin{eqnarray*}
P_{\rm Mat}({\cal M},{\cal M}')=\left\{ \begin{array}{ll}
1/2m & ({\rm if}\ {\cal M}'={\cal M}(e)) \\
1-|{\cal N}_{\rm Mat}({\cal M})|/2m & ({\rm if}\ {\cal M}'={\cal M}) \\
0 & ({\rm otherwise}) \\ 
\end{array} \right.
\end{eqnarray*}
for any ${\cal M},{\cal M}'\in \Omega_{\rm Mat}$. 
Note that $P_{\rm Mat}$ is ergodic and reversible, and 
  its stationary distribution is uniform on $\Omega_{\rm Mat}$~\cite{JS96}. 
The following theorem is due to Jerrum and Sinclar~\cite{JS96}.  
\begin{theorem}
\cite{JS96}
\label{thm:matchmix}
For $P_{\rm Mat}$, 
\begin{eqnarray*}
\tau (\gamma )\leq 4mn(n\ln n +\ln \gamma ^{-1})
\end{eqnarray*}
holds for any $\gamma >0$. 
\end{theorem}

 It is not difficult to see that the maximum degree $\Delta = m+1$ (including a self-loop) 
   of the transition diagram $P_{\rm LIN}$. 
 Thus, Theorem~\ref{thm:mixupper-vertexbi2} suggests the following\footnote{
   Theorem~\ref{thm:mixupper-vertexgreedy} also suggests that 
$|\chi_w^{(T)}-\mu_w^{(T)} | \leq 4(m+1)mn(n\ln n +\ln 4) = \Order(m^2 n^2 \log n)$ 
for the SRT-router model.  }. 
\begin{theorem}\label{thm:matchingupper}
For the billiard-router model corresponding to $P_{\rm LIN}$, 
 the discrepancy between $\chi^{(T)}$ and $\mu^{(T)}$ satisfies
\begin{eqnarray*}
\left|\chi_w^{(T)}-\mu_w^{(T)}\right| 
\leq 4m^2 n(n\ln n +\ln 4)
= \Order(m^2 n^2 \log n)
\end{eqnarray*}
for any $w\in V$ and $T\geq 0$. 
\end{theorem}

\section{Concluding Remarks}
 This paper has been concerned with the functional-router model, 
   that is a generalization of the rotor-router model, and 
   gave an upper bound of $|\chi^{(t)}_v - \mu^{(t)}_v|$ 
   when its corresponding Markov chain is reversible. 
 We can also show a similar bound 
  for a version of functional-router model 
  with oblivious routers (see \cite{shiraga}). 
 A bound of the point-wise distance 
   independent of $\pi_{\max}/\pi_{\min}$ and/or independent of $\Delta$ 
  is a future work. 
 Development of deterministic approximation algorithms 
  based on deterministic random walks for {\#}P-hard problems is a challenge.

\bibliographystyle{abbrv}

%
\appendix
\section{Fundamental Properties of Markov Chain and Mixing Time}\label{sec:RMC}
\subsection{Proof of Proposition \ref{prop:dltimes}}
In this section, we show Proposition \ref{prop:dltimes} (see e.g., ~\cite{LPW08, MT06}). 
\paragraph{Proposition \ref{prop:dltimes}}
  For any integers $\ell$  $(\ell\geq 1)$ and 
  $k$ $(0 \leq k < \tau(\gamma))$, 
\begin{eqnarray*}
h \left(\ell \cdotp \tau(\gamma)+k \right)\leq \frac{1}{2}(2\gamma )^\ell
\end{eqnarray*}
 holds for any $\gamma$ $(0<\gamma<1/2)$. 

To begin with, we define 
\begin{eqnarray}
\label{def:dbt}
\bar{h} (t)\defeq \max_{v, w\in V}\dtv(P^t_{v, \cdot}, P^t_{w, \cdot}). 
\end{eqnarray}
Then, we show the  following. 
\begin{lemma}
\label{lemm:dbarp}
Let $\xi, \zeta\in \mathbb{R}^{|V|}$ be arbitrary probability distributions. 
Then, 
\begin{eqnarray*}
\dtv(\xi P^t, \zeta P^t) \leq \bar{h}(t)
\end{eqnarray*}
holds for any $t\geq 0$. 
\end{lemma} 
\begin{proof}
By \eqref{def:TV}, 
\begin{eqnarray}
\label{eq:ddbar1}
\dtv(\xi P^t-\zeta P^t)
&=&\frac{1}{2}\left \|\sum_{v\in V}\xi_v P^t_{v, \cdot} -\sum_{w\in V}\zeta_w P^t_{w, \cdot} \right \|_{1} \nonumber \\
&=&\frac{1}{2}\left \|\sum_{v\in V}\xi_v P^t_{v, \cdot}\sum_{w\in V}\zeta_w -\sum_{w\in V}\zeta_w P^t_{w, \cdot}\sum_{v\in V}\xi_v \right \|_{1} \nonumber \\
&=&\frac{1}{2}\left \|\sum_{v\in V}\sum_{w\in V}\xi_v \zeta_w\left( P^t_{v, \cdot}- P^t_{w, \cdot}\right)\right \|_{1}
\end{eqnarray}
holds. Notice that 
  $\sum_{u\in V}\xi_u=\sum_{u\in V}\zeta_u=1$, since $\xi$ and $\zeta$ are probability distributions. 
Thus, 
\begin{eqnarray*}
(\ref{eq:ddbar1})
&\leq &\frac{1}{2}\sum_{v\in V}\sum_{w\in V}\xi_v \zeta_w\left \|P^t_{v, \cdot}- P^t_{w, \cdot}\right \|_{1} \nonumber \\
&\leq &\frac{1}{2}\max_{v, w\in V}\left \|P^t_{v, \cdot}- P^t_{w, \cdot}\right \|_{1}\sum_{v\in V}\sum_{w\in V}\xi_v \zeta_w \\
&=& \frac{1}{2}\max_{v, w\in V}\left \|P^t_{v, \cdot}- P^t_{w, \cdot}\right \|_{1}=\bar{h}(t)
\end{eqnarray*}
holds, where the second last equality follows $\sum_{v\in V}\sum_{w\in V}\xi_v \zeta_w=\sum_{v\in V}\xi_v\sum_{w\in V}\zeta_w=\sum_{v\in V}\xi_v=1$, and we obtain the claim. 
\end{proof}

\begin{lemma}
\label{lemm:ddbarineq}
\begin{eqnarray*}
h(t)\leq \bar{h}(t)\leq 2h(t)
\end{eqnarray*}
holds for any $t\geq 0$. 
\end{lemma} 
\begin{proof}
Let $\Vec e_v\in \mathbb{R}^{|V|}$ denote the $v$-th unit vector. 
By Lemma \ref{lemm:dbarp}, 
\begin{eqnarray*}
\dtv(P^t_{v, \cdot}, \pi)
=\dtv(\Vec e_v P^t, \pi P^t)
\leq \bar{h}(t) 
\end{eqnarray*}
holds for any $v\in V$, and we obtain $h(t)\leq \bar{h}(t)$. 

By the definition of the total variation distance, 
\begin{eqnarray*}
\dtv(P^t_{v, \cdot}, P^t_{w, \cdot} )
&=&\frac{1}{2}\sum_{u\in V} \left| P^t_{v, u}-\pi_u+\pi_u-P^t_{w, u}\right|\\
&\leq &\frac{1}{2}\sum_{u\in V} \left| P^t_{v, u}-\pi_u\right| +\frac{1}{2}\sum_{u\in V} \left| \pi_u-P^t_{w, u}\right|
\leq  2h(t)
\end{eqnarray*}
holds for any $v, w\in V$. We obtain $\bar{h}(t)\leq 2h(t)$. 
\end{proof}

\begin{lemma}
\label{lemm:dbarineq}
Suppose a vector $\xi\in \mathbb{R}^{|V|}$ satisfies $\sum_{i\in V}\xi_i=0$, then
\begin{eqnarray*}
\left \|\xi P^t \right \|_{1}\leq \left \|\xi \right \|_{1}\bar{h}(t)
\end{eqnarray*}
holds for any $t\geq 0$. 
\end{lemma} 
\begin{proof}
For convenience, let $\xi^+, \xi^-\in \mathbb{R}^{|V|}$ be defined by $\xi^+_i=\max\{\xi_i, 0\}$ and $\xi^-_i=\max\{-\xi_i, 0\}$. 
Then, 
\begin{eqnarray}
\label{eq:num}
\xi^+_i-\xi^-_i=\max\{\xi_i, 0\}-\max\{-\xi_i, 0\}=\xi_i
\end{eqnarray}
holds, meaning that 
\begin{eqnarray}
\label{eq:nupnum}
\xi=\xi^+-\xi^-. 
\end{eqnarray}
Since $\sum_{i\in V}\xi_i=0$, 
\begin{eqnarray}
\label{eq:nueqpm}
\sum_{i\in V}\xi^+_i=\sum_{i\in V}\xi^-_i
\end{eqnarray}
holds by \eqref{eq:num}. By the definition of $\xi^+$ and $\xi^-$, 
\begin{eqnarray*}
\xi^+_i+\xi^-_i=\max\{\xi_i, 0\}+\max\{-\xi_i, 0\}=|\xi_i|
\end{eqnarray*}
holds. Hence
\begin{eqnarray}
\label{eq:npsum}
\sum_{i\in V}\xi^+_i+\sum_{i\in V}\xi^-_i=\sum_{i\in V}|\xi_i|
\end{eqnarray}
holds. 
Thus, by (\ref{eq:nueqpm}) and (\ref{eq:npsum}), 
\begin{eqnarray}
\label{eq:vpvmv}
\sum_{i\in V}\xi^+_i=\sum_{i\in V}\xi^-_i=\frac{1}{2}\sum_{i\in V}|\xi_i|=\frac{1}{2}\|\xi\|_{1}
\end{eqnarray}
holds, hence $\frac{\xi^+}{\frac{1}{2}\|\xi\|_{1}}$ and $\frac{\xi^-}{\frac{1}{2}\|\xi\|_{1}}$ are probabilistic distribution, respectively. 
Finally, by Lemma \ref{lemm:dbarp} and (\ref{eq:nupnum}), 
\begin{eqnarray*}
\left \|\xi P^t \right \|_{1}
&=&\left \|\xi^+ P^t-\xi^- P^t\right \|_{1}
=\frac{1}{2}\|\xi\|_{1} \cdotp \left \| \frac{\xi^+}{\frac{1}{2}\|\xi\|_{1}} P^t-\frac{\xi^-}{\frac{1}{2}\|\xi\|_{1}} P^t \right \|_{1}\\
&\leq &\left \|\xi \right \|_{1}\bar{h}(t)
\end{eqnarray*}
 holds, and we obtain the claim. 
\end{proof}

\begin{lemma}
\label{lemm:smp}
\begin{eqnarray*}
h(s+t) &\leq& h(s)\bar{h}(t), \quad\mbox{and}\\
\bar{h}(s+t)&\leq& \bar{h}(s)\bar{h}(t)  
\end{eqnarray*}
hold for any $s, t\geq 0$. 
\end{lemma} 
\begin{proof}
By Lemma \ref{lemm:dbarineq}, 
\begin{eqnarray*}
\frac{1}{2}\left \| \Vec e_v P^{s+t}-\pi\right\|_{1}
=\frac{1}{2}\left \| \left( \Vec e_v P^{s}-\pi \right) P^t\right\|_{1}
&\leq &\frac{1}{2}\left \| \Vec e_v P^{s}-\pi \right\|_{1}\bar{h}(t)
\end{eqnarray*}
holds for any $v\in V$, and we get $h(s+t)\leq h(s)\bar{h}(t)$. 
Similarly, 
\begin{eqnarray*}
\frac{1}{2}\left \| \Vec e_v P^{s+t}-\Vec e_w P^{s+t} \right\|_{1}
=\frac{1}{2}\left \| \left( \Vec e_v P^{s}-\Vec e_w P^{s} \right) P^t\right\|_{1}
&\leq &\frac{1}{2}\left \| \Vec e_v P^{s}-\Vec e_w P^{s} \right\|_{1}\bar{h}(t)
\end{eqnarray*}
holds for any $v, w\in V$, and we get $\bar{h}(s+t)\leq \bar{h}(s)\bar{h}(t)$. 
\end{proof}

\begin{proof}[Proof of Proposition \ref{prop:dltimes}]
Using Lemma \ref{lemm:smp}, 
\begin{eqnarray}
h\left(\ell \cdotp \tmix(\gamma )+k \right)
&\leq &h\left(\ell \cdotp \tmix(\gamma )\right){\bar h}(k)
\leq h\left(\ell \cdotp \tmix(\gamma )\right) \nonumber \\
&\leq &h\left(\tmix(\gamma )\right) \cdotp \bar{h} \left((\ell-1)\cdotp \tmix(\gamma )\right) 
\leq h\left(\tmix(\gamma )\right) \cdotp \left(\bar{h} \left(\tmix(\gamma )\right)\right)^{\ell-1}
\label{eq:db1}
\end{eqnarray}
holds. 
By Lemma \ref{lemm:ddbarineq}, 
\begin{eqnarray*}
(\ref{eq:db1})
\leq h\left(\tmix(\gamma )\right) \cdotp \left(2 h\left(\tmix(\gamma )\right)\right)^{\ell-1}
\leq \gamma \cdotp (2\gamma )^{\ell-1}
\leq \frac{1}{2}(2\gamma)^\ell
\end{eqnarray*}
holds, and we obtain the claim. 
\end{proof}

\subsection{Supplemental proof of Theorem~\ref{thm:mixupper-vertexds}}
We show the following proposition, appearing in the proof of Theorem~\ref{thm:mixupper-vertexds}.  
\begin{proposition}
\label{prop:reversible}
If $P$ is reversible, then 
\begin{eqnarray*}
 \pi_uP^t_{u, v}=\pi_vP^t_{v, u}
\end{eqnarray*}
holds for any $u, v\in V$ and for any $t\geq 1$. 
\end{proposition}
\begin{proof}
 We show the claim by an induction of $t$. 
 For $t=1$, the claim is clear by the definition of reversible. 
 Assuming that $\pi_{u'}P^t_{u', v'}=\pi_{v'}$ holds for any $u', v'\in V$, 
 we show that $\pi_uP^{t+1}_{u, v}=\pi_vP^{t+1}_{v, u}$ holds for any $u, v\in V$ as follows;  
\begin{eqnarray*}
\pi_uP^{t+1}_{u, v}
&=&  \pi_u\sum_{w \in V}P^t_{u, w}P_{w, v}
\ =\ \sum_{w \in V}\pi_uP^t_{u, w}P_{w, v} \\
&=&  \sum_{w \in V}\pi_w P^t_{w, u}P_{w, v}
\ =\ \sum_{w \in V}P^t_{w, u} \pi_wP_{w, v}\\
&=&  \sum_{w \in V}P^t_{w, u} \pi_v P_{v, w}
\ =\ \pi_v \sum_{w \in V} P_{v, w}P^t_{w, u} \\
&=& \pi_vP^{t+1}_{v, u}. 
\end{eqnarray*}
We obtain the claim.
\end{proof}

\section{Supplemental proofs in Section~\ref{sec:routingmodel}}
\subsection{Proof of Theorem~\ref{thm:mixupper-vertexbi2}}\label{appendix:billiard}
Remark that
\begin{eqnarray}
\label{eq:reverseP}
\sum_{t=0}^{T-1} \sum_{u\in V} \left| P^{T-t-1}_{u, w}-\pi_w \right|
&\leq & \frac{\pi_w}{\pi_{\min}} \sum_{t=0}^{T-1} \sum_{u\in V} \left| P^{T-t-1}_{w, u}-\pi_u \right|
\end{eqnarray}
holds using Proposition~\ref{prop:reversible} (cf. recall that the arguments on \eqref{eq:mdisc4} in Section~\ref{sec:main}). 
We also remark that
\begin{eqnarray}
\label{eq:mixbi}
\sum_{t=0}^{T-1} \sum_{u\in V} \left| P^{T-t-1}_{w, u}-\pi_u \right|=\sum_{t=0}^{T-1} 2\dtv (P^t_{w,\cdot},\pi)\leq \frac{2(1-\gamma)}{1-2\gamma}\tau( \gamma)=3t^*
\end{eqnarray}
holds by Lemma~\ref{lemm:dtsum}. 
\begin{proof}[Proof of Theorem~\ref{thm:mixupper-vertexbi2}]
By (\ref{eq:mdisc1}) and Lemma~\ref{bound:billiardz}, we obtain 
\begin{eqnarray}
\label{eq:apbi1}
\lefteqn{ \left| \chi ^{(T)}_w-\mu ^{(T)}_w\right| \leq  \sum_{t=0}^{T-1} \sum_{u\in V} \sum_{v\in {\cal N}(u)}\left| Z_{v, u}^{(t)}-\chi_v^{(t)}P_{v, u}\right|  \left| P^{T-t-1}_{u, w}-\pi_w \right| } \nonumber \\
&\leq & \sum_{t=0}^{T-1} \sum_{u\in V} \sum_{v\in {\cal N}(u)}(1+P_{v,u}(\delta^+(v)-2)) \left| P^{T-t-1}_{u, w}-\pi_w \right| \nonumber \\
&= & \sum_{t=0}^{T-1} \sum_{u\in V} \left| P^{T-t-1}_{u, w}-\pi_w \right| \sum_{v\in {\cal N}(u)}1 
+ \sum_{t=0}^{T-1} \sum_{u\in V} \left| P^{T-t-1}_{u, w}-\pi_w \right| \sum_{v\in {\cal N}(u)} P_{v,u}(\delta(v)-2). 
\end{eqnarray}
By combining (\ref{eq:reverseP}) and (\ref{eq:mixbi}), we obtain
\begin{eqnarray}
\label{eq:apbi2}
\sum_{t=0}^{T-1} \sum_{u\in V} \left| P^{T-t-1}_{u, w}-\pi_w \right| \sum_{v\in {\cal N}(u)}1 
\leq \frac{\Delta\pi_w}{\pi_{\min}} \sum_{t=0}^{T-1} \sum_{u\in V} \left| P^{T-t-1}_{w, u}-\pi_u \right|
\leq \frac{3\pi_w}{\pi_{\min}}t^*\Delta. 
\end{eqnarray}
Since $P$ is reversible, we obtain 
\begin{eqnarray}
\label{eq:apbi3}
\sum_{v\in {\cal N}(u)} P_{v,u}(\delta(v)-2) 
\leq (\Delta-2)\sum_{v\in {\cal N}(u)} P_{v,u}
=(\Delta-2)\sum_{v\in {\cal N}(u)} \frac{\pi_u P_{u,v}}{\pi_v}
\leq (\Delta-2) \frac{\pi_u}{\pi_{\min}}, 
\end{eqnarray}
and hence
\begin{eqnarray}
\label{eq:apbi4}
\lefteqn { \sum_{t=0}^{T-1} \sum_{u\in V} \left| P^{T-t-1}_{u, w}-\pi_w \right| \sum_{v\in {\cal N}(u)} P_{v,u}(\delta(v)-2)
\leq \frac{\Delta -2}{\pi_{\min}} \sum_{t=0}^{T-1} \sum_{u\in V} \left| P^{T-t-1}_{u, w}-\pi_w \right| \pi_u } \nonumber \\
&=& \frac{\Delta -2}{\pi_{\min}} \sum_{t=0}^{T-1} \sum_{u\in V} \left| \pi_u P^{T-t-1}_{u, w}-\pi_u\pi_w \right| 
= \frac{\Delta -2}{\pi_{\min}} \sum_{t=0}^{T-1} \sum_{u\in V} \left| \pi_w P^{T-t-1}_{w, u}-\pi_u\pi_w \right| \nonumber \\
&=&\frac{(\Delta -2)\pi_w}{\pi_{\min}} \sum_{t=0}^{T-1} \sum_{u\in V} \left| P^{T-t-1}_{w, u}-\pi_u \right| 
\leq \frac{3\pi_w}{\pi_{\min}}t^*(\Delta-2). 
\end{eqnarray}
Thus, we obtain the claim by (\ref{eq:apbi1}), (\ref{eq:apbi2}) and (\ref{eq:apbi4}). 
\end{proof}

\subsection{Supplemental Proofs in Section~\ref{sec:vander}}\label{appendix:vander}
 This section presents a proof of Theorem~\ref{thm:upper-logM} and Lemma~\ref{bound:vandercz}. 
To begin with, we remark two lemmas concerning the function $\psi$ defined by~(\ref{eq:def-psi}). 
\begin{lemma}\label{lemmf2}
 For any $i\in \mathbb{Z}_{\geq 0}$ and any $k \in \mathbb{Z}_{\geq 0}$, 
\begin{eqnarray*}
\psi(i)=\psi\left( i \bmod{2^k} \right) +\psi\left( \left \lfloor \frac{i}{2^k}\right \rfloor \right)\cdotp \frac{1}{2^k}
\end{eqnarray*}
 holds. 
\end{lemma}
\begin{proof}
 In case of $i<2^k$, 
   the claim is easy 
  since
   $\psi(i \bmod{2^k}) = \psi(i)$ and 
   $\psi(\lfloor i/2^k \rfloor )=\psi(0)=0$ holds. 
 Suppose that $i\geq 2^k$, and that 
  $i$ 
   is represented in binary as $i=\sum^{\lfloor \lg i \rfloor}_{j=0} \beta_j(i) \cdotp 2^j$ 
   using $\beta_j(i) \in \{0,1\}$ ($j \in \{0,1,\ldots, \lfloor \lg i \rfloor\}$). 
 Then, 
\begin{eqnarray*}
i \bmod 2^k
 &=& \left(\sum^{\lfloor \lg i \rfloor}_{j=0}\beta_j (i)  \cdotp 2^j \right) \bmod{2^k}
\ =\  \sum_{j=0}^{k-1}\beta_j (i) \cdotp 2^j, \hspace{1em} \mbox{ and}\\
\left \lfloor \frac{i}{2^k}\right \rfloor
 &=& \left \lfloor \frac{\sum^{\lfloor \lg i \rfloor}_{j=0} \beta_j (i)  \cdotp 2^j}{2^k}\right \rfloor
\ =\ \sum_{j=k}^{\lfloor \lg i \rfloor}\beta_j (i) \cdotp 2^{j-k}
\end{eqnarray*} 
  hold, respectively. 
 Let $l =j-k$ and 
 let $b_l = \beta_{l+k} (i)$ ($l = 0, 1, \ldots, \lfloor \lg i \rfloor-k)$, for convenience, 
  then 
\begin{eqnarray*}
\lefteqn{ \psi\left( i\bmod 2^k\right) +\psi\left( \left \lfloor \frac{i}{2^k}\right \rfloor \right)\cdotp\frac{1}{2^k}
 = \psi\left( \sum_{j=0}^{k-1}\beta_j (i)2^j\right) +\psi\left( \sum_{j=k}^{N}\beta_j (i) \cdotp 2^{j-k}\right) \cdotp\frac{1}{2^k} } \\ 
 &=& \sum_{j=0}^{k-1}\beta_j (i)2^{-(j+1)}+\psi\left( \sum_{l=0}^{\lfloor \lg i \rfloor-k}b_l \cdotp 2^{l}\right) \cdotp\frac{1}{2^k} 
 = \sum_{j=0}^{k-1}\beta_j (i)2^{-(j+1)}+\frac{1}{2^k} \sum_{l=0}^{\lfloor \lg i \rfloor-k}b_l \cdotp 2^{-(\ell +1)} \\
 &=& \sum_{j=0}^{k-1}\beta_j (i)2^{-(j+1)}+\frac{1}{2^k} \sum_{j=k}^{\lfloor \lg i \rfloor}\beta_j (i) \cdotp 2^{-(j-k+1)}
 = \sum_{j=0}^{k-1}\beta_j (i)2^{-(j+1)}+\sum_{j=k}^{\lfloor \lg i \rfloor}\beta_j (i) \cdotp 2^{-(j+1)}\\
 &=& \sum_{j=0}^{\lfloor \lg i \rfloor}\beta_j (i) \cdotp 2^{-(j+1)} 
 = \psi(i)
\end{eqnarray*}
 and, we obtain the claim.
\end{proof}
\begin{lemma}\label{lemmf1}
 For any $k \in \mathbb{Z}_{\geq 0}$ and $\alpha \in \{0, 1, \ldots, 2^k-1\}$, 
\begin{eqnarray*}
\psi(2^k+\alpha )=\frac{1}{2^{k+1}} +\psi(\alpha )
\end{eqnarray*}
holds. 
\end{lemma}
\begin{proof}
 By Lemma\ref{lemmf2}, 
\begin{eqnarray*}
 \psi(2^k+\alpha )
 &=&  \psi\left( (2^k+\alpha ) \bmod 2^k\right) 
     +\psi\left( \left \lfloor \frac{2^k+\alpha }{2^k}\right \rfloor \right)\cdotp \frac{1}{2^k} \\
 &=& \psi(\alpha )+\frac{\psi(1)}{2^k} \\
 &=& \frac{1}{2^{k+1}} +\psi(\alpha) 
\end{eqnarray*}
  holds, where the last equality follows  $\psi(1)=1/2$ by the definition. 
\end{proof}
 Now, we define 
\begin{eqnarray}
 \Phi[z, z') \defeq \{\psi(i)\mid i \in \{ z, z+1, \ldots, z'-1 \} \}
 \label{def:Psi}
\end{eqnarray}
  for $z,z' \in \mathbb{Z}_{\geq 0}$ satisfying $z < z'$. 
 For convenience, we define $\Phi[z, z) = \emptyset$. 
 It is not difficult to see that 
\begin{eqnarray}
 \Phi [0, 2^k)=\left \{ \frac{i}{2^{k}}\mid i \in \{0, 1, \ldots, 2^k-1 \} \right \}
\label{eq:coloF}
\end{eqnarray}
 holds for any $k \in \mathbb{Z}_{\geq 0}$. 
 In general, we can show the following lemma, using Lemmas~\ref{lemmf2} and \ref{lemmf1}. 
\begin{lemma}\label{lemmfxy2}
For any $z \in \mathbb{Z}_{\geq 0}$ and for any $k \in \mathbb{Z}_{\geq 0}$
\begin{eqnarray*}
\left| \Phi[z, z+2^k)\cap \left[ \frac{i}{2^k}, \frac{i+1}{2^k}\right) \right|=1
\end{eqnarray*}
 holds for any $i\in \{0, 1, \ldots, 2^k-1\}$. 
\end{lemma}
\begin{proof}
 By Lemma~\ref{lemmf2}, 
\begin{eqnarray*}
\Phi[z, z+2^k) 
 &=& \left\{\psi(z+i) \mid i \in \{0,1,\ldots,2^k-1\} \right\}  \\
 &=& \left\{\psi\left((z+i) \bmod 2^k\right) +\frac{\psi\left(\left\lfloor \frac{z+i}{2^k} \right\rfloor\right)}{2^k}
  \mid i \in \{0,1,\ldots,2^k-1\} \right\}  
\end{eqnarray*}
  holds. 
 Since $0 \leq \psi(z') < 1$ holds for any $z' \in \mathbb{Z}_{\geq 0}$, 
\begin{eqnarray}
\lefteqn{\psi\left( (z + i) \bmod 2^k\right) + \frac{\psi\left( \left \lfloor \frac{z+i}{2^k}\right \rfloor \right)}{2^k}} \nonumber\\
&\in& \left[\psi\left( (z + i) \bmod 2^k\right),\psi\left( (z + i) \bmod 2^k\right) +\frac{1}{2^k} \right)
\label{eq:lemmfxy2-a}
\end{eqnarray}
  holds for each $i \in \{0,1,\ldots,2^k-1\}$. 
 The observation (\ref{eq:coloF}) implies that 
\begin{eqnarray}
\lefteqn{ \left \{ \psi\left( (z + i) \bmod 2^k\right) \mid i \in \{0,1,\ldots,2^k-1\} \right\}  } \nonumber\\
&=& 
 \left \{ \frac{j}{2^k} \mid j \in \{0,1,\ldots,2^k-1\} \right\}  
\label{eq:lemmfxy2-b}
\end{eqnarray}
  holds. 
 Notice that (\ref{eq:lemmfxy2-b}) implies 
 $\psi\left( (z + i) \bmod 2^k\right) \neq \psi\left( (z + i') \bmod 2^k\right) $
  for any distinct $i,i' \in  \{0,1,\ldots,2^k-1\} $. 
 Now, the claim is clear by (\ref{eq:lemmfxy2-a}) and (\ref{eq:lemmfxy2-b}). 
\end{proof}

 Note that Lemma~\ref{lemmfxy2} implies that 
\begin{eqnarray}
 \Phi [z, z+2^k) 
 = \left \{ \frac{i}{2^{k}} + \frac{\theta(i)}{2^{k}} \mid i \in \{0, 1, \ldots, 2^k-1 \} \right \}
\label{eq:FPSI-explicit}
\end{eqnarray}
  using appropriate $\theta(i) \in [0,1)$ for $i = 0,1,\ldots,2^k-1$. 

\begin{lemma}\label{lemmIvu1}
 Let $z_0,k \in \mathbb{Z}_{\geq 0}$, and let $x,y \in [0,1)$ satisfy $x<y$. 
 Then, 
\begin{eqnarray*}
\left|  | \Phi[z_0, z_0+2^k) \cap [x,y) |  - 2^k \cdotp (y-x) \right| < 2 
\end{eqnarray*}
holds. 
\end{lemma}
\begin{proof}
 Lemma~\ref{lemmfxy2} and Equation~(\ref{eq:FPSI-explicit}) implies that 
  there exists $s, t \in \{0,1,\ldots,2^k-1\}$ such that $s \leq t$ and 
\begin{eqnarray}
 \frac{s}{2^k}+\frac{\theta(s)}{2^k} \leq &x& < \frac{s+1}{2^k}+\frac{\theta(s+1)}{2^k} \label{eq4} \\
 \frac{t}{2^k}+\frac{\theta(t)}{2^k} \leq &y& < \frac{t+1}{2^k}+\frac{\theta(t+1)}{2^k} \label{eq5} 
\end{eqnarray}
 where we assume 
\begin{eqnarray*}
 \frac{-1}{2^k}+\frac{\theta(-1)}{2^k} = 0 &\mbox{and}& 
 \frac{2^k}{2^k}+\frac{\theta(2^k)}{2^k} = 1 
\end{eqnarray*}
  for convenience. 
 Then, it is clear that 
\begin{eqnarray*}
 \Phi[z_0, z_0+2^k) \cap [x,y) 
 &=& \left \{ \frac{i}{2^k}+\frac{\theta(i)}{2^k} \mid i \in \{s, s+1, \ldots, t-1 \} \right \}
\end{eqnarray*} 
 where $s=t$ means that $\Phi[z_0, z_0+2^k) \cap [x,y) = \emptyset$. 
 By (\ref{eq4}) and (\ref{eq5}), we obtain 
\begin{eqnarray*}
 2^k \cdotp x - \theta(s+1)-1  < &s& \leq 2^k \cdotp x -\theta(s), \\
 2^k \cdotp y - \theta(t+1)-1  < &t& \leq 2^k \cdotp y -\theta(t)
\end{eqnarray*}
 respectively, and hence we obtain that 
\begin{eqnarray*}
2^k \cdotp (y-x) -\theta(t+1)-1+\theta(s) \ <\ 
t-s 
\ <\ 2^k \cdotp (y-x) -\theta(t)+\theta(s+1)+1. 
\end{eqnarray*}
 Since $| \Phi[z_0, z_0+2^k) \cap [x,y)| = t-s$ and 
  $0 \leq \theta(i) < 1$ ($i \in \{0,1,\ldots,2^k-1\}$), 
\begin{eqnarray*}
 2^k \cdotp (y-x) -2 < \left|\Phi[z_0, z_0+2^k) \cap [x,y) \right| < 2^k \cdotp(y-x) +2
\end{eqnarray*}
 holds, and we obtain the claim. 
\end{proof}
 Using Lemma~\ref{lemmIvu1}, 
  we obtain the following.  
\begin{lemma}\label{theoIvu3}
 Let $z_0 \in \mathbb{Z}_{\geq 0}$, $z \in \mathbb{Z}_{> 0}$, and let $x,y \in [0,1)$ satisfy $x<y$. 
 Then, 
\begin{eqnarray*}
\left|\frac{\left| \Phi[z_0, z_0+z) \cap [x,y)\right|}{z} - (y-x) \right| < \frac{2 \lfloor \lg z \rfloor + 2}{z} 
\end{eqnarray*}
  holds. 
\end{lemma}
\begin{proof}
 For simplicity, 
  let 
\begin{eqnarray*}
  \Phi^*[z_0, z_0+ \ell ) \defeq \Phi[z_0, z_0+ \ell ) \cap [x,y)  
\end{eqnarray*}
 for any $ \ell  \in \mathbb{Z}_{>0}$. 
 Then, notice that 
\begin{eqnarray}
|\Phi^*[z_0, z_0+z)| 
= \left|\Phi^*[z_0, z_0+z') \right| + \left|\Phi^*[z_0+z', z_0+z) \right|
\label{eq:sep-Psi}
\end{eqnarray}
  holds for any $z' \in \mathbb{Z}_{> 0}$ satisfying $z' < z$. 
 Now, suppose $z$ is represented as 
  $z=\sum_{j=0}^{\lfloor \lg y \rfloor} \beta_j(z) \cdotp 2^j$ in binary, 
  where $\beta_j(z) \in \{0,1\}$. 
 Using Lemma~\ref{lemmIvu1}, we obtain that 
\begin{eqnarray*}
|\Phi^*[z_0, z_0+z)|&=& 
  \left|\Phi^*\left[z_0, z_0+\sum_{j=0}^{\lfloor \lg z \rfloor} \beta_j(z) \cdotp 2^j \right)\right| \\
 &=& 
    \left|\Phi^*\left[ z_0, z_0+\beta_0(z) \cdotp 2^0\right) \right| 
    +\sum_{k=0}^{\lfloor \lg z \rfloor-1} 
     \left|\Phi^*\left[ z_0+\sum_{j=0}^{k}\beta_j(z) \cdotp 2^j, z_0+\sum_{j=0}^{k+1}\beta_j(z) \cdotp 2^j\right) \right|\\
 &<& \beta_0(z) \cdotp 2^0 \cdotp (y-x)+2 
    +\sum_{k=0}^{\lfloor \lg z \rfloor-1}  \left( \beta_{k+1}(z) \cdotp 2^{k+1} \cdotp (y-x) + 2 \right) \\
 &=& \sum_{k=0}^{\lfloor \lg z \rfloor}  \left( \beta_k(z) \cdotp 2^k \cdotp (y-x) + 2 \right) \\
 &=& 2 (\lfloor \lg z \rfloor+1) + (y-x) \sum_{k=0}^{\lfloor \lg z \rfloor} \beta_k(z) \cdotp 2^k\\
 &=& 2 (\lfloor \lg z \rfloor+1) + z \cdotp (y-x). 
\end{eqnarray*}
 In a similar way, 
  we also have 
\begin{eqnarray*}
|\Phi^*[z_0, z_0+z)| > 2 (\lfloor \lg z \rfloor-1) + z \cdotp (y-x), 
\end{eqnarray*}
 and we obtain the claim. 
\end{proof}
 By Lemma~\ref{theoIvu3}, it is not difficult to see 
 Theorem~\ref{thm:upper-logM} and Lemma~\ref{bound:vandercz} holds. 

\subsection{Proof of Theorem~\ref{thm:mixupper-vertexrr2}}\label{appendix:rotor}
By (\ref{eq:mdisc1}) and Observation~\ref{bound:rotorz}, we obtain 
\begin{eqnarray}
\label{eq:appro1}
\left| \chi ^{(T)}_w-\mu ^{(T)}_w\right| 
&\leq &\sum_{t=0}^{T-1} \sum_{u\in V} \sum_{v\in {\cal N}(u)}\left| Z_{v, u}^{(t)}-\chi_v^{(t)}P_{v, u}\right|  \left| P^{T-t-1}_{u, w}-\pi_w \right| \nonumber \\
&\leq &\sum_{t=0}^{T-1} \sum_{u\in V} \left| P^{T-t-1}_{u, w}-\pi_w \right| \sum_{v\in {\cal N}(u)}{\bar \delta}(v)P_{v,u} \nonumber \\
&\leq &\sum_{t=0}^{T-1} \sum_{u\in V} \left| P^{T-t-1}_{u, w}-\pi_w \right| \sum_{v\in {\cal N}(u)}{\bar \delta}(v)\frac{\pi_u P_{u,v}}{\pi_v} \nonumber \\
&\leq &\frac{\max_{v\in V}{\bar \delta}(v)}{\pi_{\min}}\sum_{t=0}^{T-1} \sum_{u\in V} \left| \pi_u P^{T-t-1}_{u, w}-\pi_u \pi_w \right| \sum_{v\in {\cal N}(u)}P_{u,v} \nonumber \\
&=&\frac{\max_{v\in V}{\bar \delta}(v)\pi_w}{\pi_{\min}}\sum_{t=0}^{T-1} \sum_{u\in V} \left| P^{T-t-1}_{w,u}-\pi_u \right|. 
\end{eqnarray}
Thus, we obtain the claim by combining (\ref{eq:mixbi}) and (\ref{eq:appro1}).

\end{document}